%% file: paper.tex
\documentclass[11pt]{article}

\usepackage{packages}
\usepackage{theorems}
\usepackage{notation}

\usetikzlibrary{arrows.meta, patterns.meta}
\pgfplotsset{compat=1.18}

\input{figures/strong-net-figure-styles}

\usepackage{color-edits}
\addauthor{g}{blue}
\addauthor{u}{magenta}

\newcommand*{\pF}{\mathsf{F}}
\newcommand*{\pR}{\mathsf{R}}
\newcommand*{\pP}{\mathsf{P}}
\newcommand*{\Eof}[1]{\mathring{\mathsf E}\paren{#1}}
\newcommand*{\Wof}[1]{\mathring{\mathsf W}\paren{#1}}
\newcommand*{\Nof}[1]{\mathring{\mathsf N}\paren{#1}}
\newcommand*{\Sof}[1]{\mathring{\mathsf S}\paren{#1}}
\DeclareMathOperator{\plu}{plu}
\DeclareMathOperator{\West}{West}
\DeclareMathOperator{\East}{East}
\DeclareMathOperator{\North}{North}
\DeclareMathOperator{\South}{South}

\title{Condorcet-Winning Sets and Peer Selection\\ in Planar Metric Elections}
\author{Gabriel de Azevedo\footnote{Cornell University. Email: gm587@cornell.edu} \  and Ulysse Hennebelle\footnote{Cornell Tech, Cornell University. Email: uh34@cornell.edu}}

\begin{document}
\maketitle

\begin{abstract}
	In ranked-choice voting, a Condorcet-winning set is a group of candidates for which no outside candidate is preferred to every member of the group by a majority of voters.
	We study Condorcet-winning sets in planar metric elections, where voters rank candidates according to their distance under a given norm. We formulate general metric elections, peer selection, and the one-round Voronoi game as instances of a two-player Stackelberg game and place these problems in a common hierarchy. We also introduce a new variant, which we call \emph{strong peer selection}.

	Our main result concerns peer selection under the $\ell_1$ and $\ell_\infty$ norms. 
	We prove that every planar instance admits a Condorcet-winning set of size at most three, even under strong peer selection.
	This follows from a new result for strong rectangular $\varepsilon$-nets. We show that, for every set of points in the plane, one can choose at most three input points that intersect every axis-parallel rectangle containing more than half of the points, improving the previous threshold of $9 / 16$ due to \citet{AshokAG14}.
	Under the $\ell_2$ norm, we prove that every planar metric election admits a Condorcet-winning set of size at most four, improving the general bound of five due to \citet{SongNL26}. 
    Finally, we give new norm-independent bounds for the one-round Voronoi game.

\end{abstract}

\section{Introduction}
Legislatures, juries, boards, and committees are practical cases in which a task is delegated to a group rather than a single elected individual. 
Beyond the idea that deliberation may improve the quality of a decision, a group can satisfy desiderata that a single individual cannot.
In legislatures and juries, representativeness is central. A multi-member body can reflect a broader range of opinions and viewpoints than any one of its members alone.
Another desirable property is stability, which asks that no outside candidate should be preferred to every elected member by a strict majority of voters.
When the body consists of a single member, this is the same as electing a candidate who defeats every other candidate in a pairwise majority vote.
In 1785, Condorcet showed that such a candidate may not exist~\cite{Condorcet85}.
Whenever one does, it is called a \emph{Condorcet winner}.
This impossibility result does not make this majority guarantee less desirable.
Instead, it suggests that selecting a single candidate may be too restrictive.
\citet{Elkind15} proposed extending the requirement from a candidate to a set of candidates.
A set of candidates is a \emph{Condorcet-winning set} if, for every candidate outside the set, at most half of the voters prefer that candidate to every member of the set.
In arbitrary elections, two candidates may not suffice, whereas five candidates always do~\cite{Elkind15,SongNL26}.

One place where stability is quite valuable is in routine collective decisions, such as when a neighborhood association elects representatives from its residents or a department forms a committee from its faculty.
These decisions occur far more frequently and shape much of the day-to-day governance of communities and organizations. 
These are instances of \emph{peer selection}, where the voters form the set of eligible candidates~\cite{AlonFPT11,AzizLMRW16}. 
Because these decisions often involve small working groups, the minimum size needed to guarantee a majority is even more relevant.

In all those cases, there is a question of how preferences are generated.
We study \emph{spatial models of voting}, in which voters and candidates are represented as points in a policy space, and voters prefer candidates who are closer to them~\cite{Hotelling29,Black48}.
The plane is the simplest spatial model in which two issue dimensions can vary independently.
It can represent, for example, an economic dimension together with a social or sociocultural dimension, allowing voters who disagree on one set of issues to agree on another.
The choice of metric determines how disagreement across these dimensions is aggregated and can therefore induce different rankings for the same placement of voters and candidates.
We focus on the $\ell_p$ norms, particularly $\ell_1$, $\ell_2$, and $\ell_\infty$.
Under the $\ell_1$ norm, disagreements along the two dimensions add linearly.
Under the $\ell_2$ norm, a large disagreement along one dimension carries more weight than the same total disagreement distributed more evenly across the dimensions.
Under the $\ell_\infty$ norm, distance is determined entirely by the dimension of greatest disagreement, and smaller disagreements along the other dimension do not affect the ranking.
Thus, even on the same policy space, the choice of norm captures different ways in which voters aggregate disagreement.

On the technical side, the planar setting connects these questions to established problems in discrete and computational geometry.
The theory of $\varepsilon$-nets asks whether a small set of points can intersect every sufficiently large member of a prescribed geometric family~\cite{HausslerW87}.
\citet{BanikJCMS16} used weak $\varepsilon$-nets to derive bounds for a continuous extension of the Condorcet-winning-set problem, known as the one-round Voronoi game.
First, a set of voters is placed in the plane.
One player then places a group of candidates anywhere in the plane, after which a second player places a challenger and tries to attract as many voters as possible.
The main difference between the Voronoi game and the Condorcet-winning-set problem lies in the sets from which these candidates may be chosen.
In the Voronoi game, both the selected group and the challenger may be placed freely.
In the Condorcet-winning-set problem, both must come from a prescribed candidate set; in peer selection, they must come from the voters themselves.
We introduce \emph{strong peer selection}, in which the selected group must consist of voters, as in peer selection, but the challenger may be placed freely, as in the Voronoi game.
Thus, strong peer selection asks whether a feasible group of peers can withstand the stronger challenge of an arbitrary point in the plane.
By varying these restrictions, we place the different problems in a common hierarchy.

These restrictions also affect how $\varepsilon$-nets can be used.
In the Voronoi game, the points of a weak $\varepsilon$-net can be used directly as the selected candidates.
When the group must instead come from a prescribed candidate set or from the voters, this connection is less direct.
We show that weak and strong $\varepsilon$-nets can be used in these different settings, depending on the geometry of the norm, and thereby provide a common approach to the planar election problems studied in this paper.

\subsection{Our Results}

We obtain new bounds on the minimum size of Condorcet-winning sets in general metric elections, peer selection, the one-round Voronoi game, and strong peer selection.
We also organize these models in a hierarchy; see \cref{fig:hierarchy}.

\medskip
\noindent
\textbf{Three voters suffice for planar peer selection under the $\ell_1$ and $\ell_\infty$ norms. } Our main result is that, in every planar peer-selection instance, there is a coalition of at most three peers such that no challenger, even one placed freely anywhere in the plane, is preferred to every member of the coalition by a majority.
Thus, three peers suffice not only for peer selection, but also for strong peer selection; see Theorem~\ref{thm:3_suffice_strong_peer}.
The proof relies on a new connection between peer selection and strong $\varepsilon$-nets for axis-parallel rectangles. We prove that, for every point set in the plane, at most three of its points can be chosen to intersect every axis-parallel rectangle containing more than half of the set. This improves the previous upper bound of $9 / 16$~\cite{AshokAG14}; see \cref{thm:strong_R_3_nets}. We also improve several lower bounds for strong nets with respect to axis-parallel rectangles; see \cref{tab:strong_rectangular_bounds}.

\medskip
\noindent
\textbf{Four candidates suffice for planar metric elections under the $\ell_2$ norm. }
Under the $\ell_2$ norm, we show that every planar metric election admits a Condorcet-winning set with at most four candidates. This improves on the general bound of five candidates~\cite{SongNL26}; see Theorem~\ref{thm:4_suffice_euclidean}. 
This gives a first answer to a question of \citet{LassotaVV26} on whether better bounds for the $\ell_2$ norm exist on the plane.

\medskip
\noindent
\textbf{Exact bounds for a single candidate. } Under the $\ell_2$ norm, we show that the worst-case fraction of voters that a challenger can capture against a single defending candidate is $2 / 3$ in the Voronoi game, peer selection, and general metric elections; see \cref{prop:1_euclidean_equiv}. Under the $\ell_1$ and $\ell_\infty$ norms, this fraction is $3 / 4$ in peer selection, strong peer selection, and general metric elections; see \cref{prop:1_manhattan_equiv}.

\medskip
\noindent
\textbf{Norm-independent bounds. }
For the Voronoi game, where candidates can be freely placed on the plane, we prove that, under any norm on the plane, four freely placed candidates can prevent every challenger from attracting a strict majority. This generalizes the result of \citet{BergW23} for the $\ell_2$ norm.
We also show that two candidates suffice when a voter is equidistant from every candidate.

\medskip
This work is organized as follows. 
First, we derive improved upper and lower bounds for strong rectangular nets.
Second, we formulate the metric election problems as a two-player Stackelberg game and build a hierarchy of these problems.
We use the results from those two sections to derive stronger bounds on the sizes of Condorcet-winning sets on the plane under multiple norms.
Finally, we derive the norm-independent bound of two candidates when all candidates are equidistant from some voter.

\subsection{Related work}

\citet{Elkind15} introduced the notion of Condorcet-winning sets, proved a logarithmic upper bound on their size, and constructed elections in which three candidates are necessary for the general case.
Later, the work of \citet{JiangMW20} on approximately stable committee selection implied a constant upper bound.
Building on their approach, \citet{CharikarLRVW25} reduced the bound to six, and \citet{SongNL26} subsequently reduced it to five.
More recently, the computational results of \citet{ZilbersteinBLM26} point to four candidates being sufficient.
Closest to our metric results, \citet{LassotaVV26} proved that four\footnote{An earlier version of the work claimed a bound of three; the current arXiv version corrects the bound to four.} candidates suffice for planar elections under the $\ell_1$ and $\ell_\infty$ norms.

Peer selection requires the selected alternatives to be drawn from the agents who evaluate them.
The algorithmic study of selection from the selectors was initiated by \citet{AlonFPT11}; related models include triadic consensus and impartial nomination~\cite{GoelL12,HolzmanM13}.
\citet{AzizLMRW16} formulated the setting explicitly as peer selection and studied strategyproof mechanisms.
To our knowledge, we are the first to study Condorcet-winning sets under the peer-selection constraint.

Early work on Voronoi games measured payoffs by the area or consumer mass captured by each player~\cite{AhnCCGO04,CheongHLM04,ChawlaRRS06}.
Discrete Voronoi games, in which payoffs are determined by a finite set of voters, were subsequently studied on the line~\cite{BanikBD13} and in two and three dimensions~\cite{BanikJCMS16}.
The latter work connected discrete Voronoi games with $\varepsilon$-nets, and \citet{BergW23} sharpened the resulting bound in the plane.
A substantial literature studies asymptotic bounds for weak $\varepsilon$-nets for convex sets~\cite{AlonBFK92,Rubin22}.
More directly related to our fixed-cardinality setting are small weak nets~\cite{AronovAHLRSS09,MustafaR09}.
We extend this connection to every norm on the plane and, under the $\ell_2$ norm, show that weak nets can be rounded to the candidate set without weakening their guarantee.

For the asymptotic theory of strong $\varepsilon$-nets, see \cite{HausslerW87,AronovES10,PachT13}.
\citet{AshokAG14} initiated the study of small strong nets for axis-parallel rectangles and proved that three input points hit every rectangle containing more than $9 / 16$ of the point set.
Our main result improves this threshold to $1 / 2$.

When the defending coalition consists of a single candidate and both players choose from the same set, the resulting minimax problem was studied by \citet{Simpson69} and \citet{Durier89}.
For freely placed candidates, the exact factors under the norms considered here were previously known~\cite{Durier89,Carrizosa96}.
We determine the exact factors for the finite-candidate and peer-restricted variants.

\section{Preliminaries}
\label{sec:preliminaries}

\noindent
\textbf{Elections.}
An \emph{election} consists of a nonempty finite set of voters $V$, a nonempty finite set of candidates $C$, and a strict total preference relation $\succ_v$ over $C$ for every voter $v \in V$.
We write $c_1 \succ_v c_2$ when voter $v$ strictly prefers candidate $c_1$ to $c_2$.
For a nonempty set $S \subseteq C$, we write $c \succ_v S$ when $v$ prefers $c$ to every candidate in $S$.
A set $S \subseteq C$ is a \emph{Condorcet-winning set} if no candidate outside $S$ is preferred to all of $S$ by a strict majority of the voters, that is,
\begin{equation*}
    \card{\setst{v \in V}{c \succ_v S}}
    \leq \frac{\card V}{2}
    \qquad \text{for every } c \in C \setminus S.
\end{equation*}
When $S$ consists of a single candidate, that candidate is a \emph{Condorcet winner}.

\medskip
\noindent
\textbf{Metric elections.}
Let $\calX$ be a finite-dimensional real vector space equipped with a metric $d$.
An election is \emph{metric} if its voters and candidates are points of $\calX$ and
\begin{equation*}
    c_1 \succ_v c_2
    \quad \Longleftrightarrow \quad
    d\paren{v, c_1} < d\paren{v, c_2}.
\end{equation*}
When the metric is induced by a norm, we denote it by the name of that norm.
Unless stated otherwise, we assume that every metric election is in general position, that is, all voters and candidates are distinct, and all preferences are strict. For $a, b \in \calX$, the \emph{dominance region} of $a$ over $b$ is the set of points of the metric space that are closer to $a$ than to $b$. We denote this region by
\begin{equation*}
    \Dom_d(a, b) \coloneqq \setst{x \in \calX}{d\paren{a, x} < d\paren{b, x}}.
\end{equation*}
When the norm is clear from the context, we omit the subscript $d$.
Under the $\ell_2$ norm, every dominance region is convex regardless of the positions of $a$ and $b$, but this is not guaranteed for arbitrary norms.
For a set $X \subseteq \calX$, we denote its boundary by $\partial X$ and its convex hull by $\conv\paren{X}$.

\medskip
\noindent
\textbf{Weak and strong $\bm{\varepsilon}$-nets.} 
Let $P \subseteq \calX$ be nonempty and finite, and let $\calK$ be a family of convex sets in $\calX$.
A \emph{weak $\varepsilon$-net of size $k$} for $P$ with respect to $\calK$ is a set $Q \subseteq \calX$ with $\card Q \leq k$ that intersects every $K \in \calK$ containing more than an $\varepsilon$ fraction of $P$.
It is a \emph{strong $\varepsilon$-net} if, in addition, $Q \subseteq P$.
For each family $\calK$ and each integer $k$, we denote the worst-case fraction of uncovered points by
\begin{equation}
    \label{eqn:weak-net-factor}
    \varepsilon_k^\calK
    \coloneqq
    \sup_{\substack{P \subseteq \calX\\ P \text{ finite}}}
    \ \inf_{\substack{Q \subseteq \calX\\ \card Q \leq k}}
    \ \sup_{\substack{K \in \calK\\ K \cap Q = \varnothing}}
    \frac{\card{K \cap P}}{\card P}.
\end{equation}
The corresponding strong factor $\varsigma_k^\calK$ is defined by additionally requiring $Q \subseteq P$ in \eqref{eqn:weak-net-factor}.
When $\calK$ is the family of all convex sets, we write $\varepsilon_k$ in place of $\varepsilon_k^\calK$.
In the plane, \emph{rectangular $\varepsilon$-nets} are a specialization in which the family $\calK$ consists of all axis-parallel rectangles.
We denote the corresponding weak and strong factors by $\varepsilon_k^{\calR}$ and $\varsigma_k^{\calR}$, respectively.

\section{\texorpdfstring{Strong Rectangular $\bm{\varepsilon}$-nets}{Strong Rectangular epsilon-nets}}
\label{sec:snets}
In the plane, for arbitrary convex sets, requiring the points of an $\varepsilon$-net to belong to the input set, that is, to form a strong net, rules out any nontrivial guarantees.
To see this, let $P$ be the vertices of a convex polygon with $n$ vertices.
For every strong net $Q \subseteq P$ of size $k$, the convex hull of the remaining points does not contain any point of $Q$.
It is therefore a convex set that avoids $Q$ and contains at least an $\paren{n - k} / n$ fraction of $P$, which approaches $1$ for fixed $k$ as $n$ increases.
Thus, there are no nontrivial bounds for strong nets for arbitrary convex sets.

Nontrivial guarantees arise when we restrict the family of sets that the net is required to intersect.
Here we consider the case where the family of convex sets consists of axis-parallel rectangles.
\citet[Lem.~7]{AshokAG14} introduced this problem and showed that three input points intersect every axis-parallel rectangle containing more than a $9 / 16$ fraction of the input points.
Our main technical result decreases this threshold to $1 / 2$.
We defer the proof details to Appendix~\ref{apx:snets} and give an overview of the main ideas here.

\begin{theorem}
    \label{thm:strong_R_3_nets}
    On the plane, $\varsigma_3^{\calR} \leq 1 / 2$.
\end{theorem}

\begin{proof}[Proof sketch.]
    Fix a finite point set $V \subseteq \R^2$.
    We call the fraction of points of $V$ contained in a set its \emph{mass}, and we call an axis-parallel rectangle \emph{dangerous} if its mass is greater than $1 / 2$.
    The proof splits the plane into a $4 \times 4$ grid in which each row and each column has mass $1 / 4$. 
    We call it the \emph{quartile grid}.
    We name the $16$ cells as in \cref{fig:grid} and call the central $2 \times 2$ block the \emph{center}.
    Up to symmetry, there are six possible cases for the center cells: all are nonempty, one is empty, two adjacent cells are empty, two diagonally opposite cells are empty, three are empty, or all are empty.
    For each case, we construct configurations of at most three input points that intersect every dangerous rectangle for every feasible distribution of mass over the quartile grid.

    \begin{figure}[ht!]
        \centering
        \input{figures/quartile-grid}
        \caption{The quartile grid and the two regions in which a dangerous rectangle avoiding a point $q$ northwest of $p$ can lie.
        Regions west and north of $q$ have mass at most $1 / 2$.}
        \label{fig:grid}
    \end{figure}
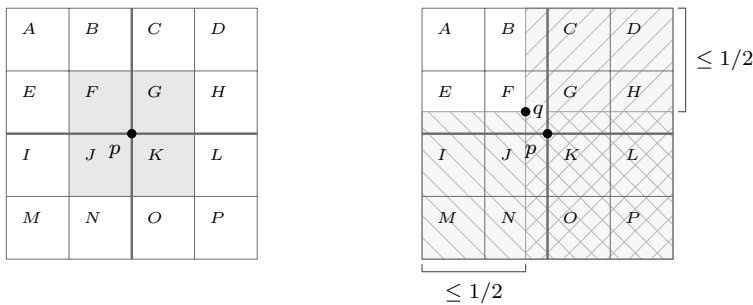

    To see why this case split makes sense, consider how a dangerous rectangle can be placed in the quartile grid.
    Let $p$ be the intersection of the vertical and horizontal middle lines.
    Each of the four half-planes determined by these lines has mass $1 / 2$, and hence every dangerous rectangle must contain $p$.
    Thus, as illustrated in \cref{fig:grid}, if a point $q$ selected for the net lies in one of the four cells northwest of $p$, namely $A$, $B$, $E$, or $F$, a dangerous rectangle avoiding $q$ must lie either south or east of $q$.
    Otherwise, it is contained in the first two rows or the first two columns and has mass at most $1 / 2$.
    Since $p$ need not belong to the input set, we check the four neighboring cells for replacements.
\end{proof}

We also improve the lower bounds for small values of $k$. 
We summarize the known upper and lower bounds for $\varsigma_k^\calR$ in \cref{tab:strong_rectangular_bounds}.
We defer the complete discussion of these lower bounds to Appendix~\ref{apx:snets}.

\begin{table}[ht!]
    \centering
    \begin{tabular}{@{}rll@{}}
        \toprule
        $k$ & Lower bound & Upper bound \\
        \midrule
        $1$ & $3 / 4$ \, \cite{AshokAG14} & $3 / 4$ \cite{AshokAG14} \\
        $2$ & $4 / 7$\,\,\, (this work) & $5 / 8$ \cite{AshokAG14} \\
        $3$ & $6 / 13$ (this work) & $1 / 2$ (this work) \\
        $4$ & $3 / 8$\,\,\, (this work) & $1 / 2$ \cite{AshokAG14} \\
        \bottomrule
    \end{tabular}
    \caption{Lower and upper bounds for strong rectangular nets.}
    \label{tab:strong_rectangular_bounds}
\end{table}

For $k = 4$, \citet[Lem.~8]{AshokAG14} previously proved that $\varsigma_4^{\calR} \leq 1 / 2$.
We believe that the factor for $k = 4$ is strictly smaller than $1 / 2$. However, our techniques target the threshold $1 / 2$ and therefore do not improve the bound for $k = 4$.
We also believe that \cref{thm:strong_R_3_nets} may be tight. Constructing a matching lower bound remains an interesting and challenging problem.

In the following sections, we show that strong rectangular nets, like weak nets, can be used to bound the minimum size of a Condorcet-winning set in the plane. To the best of our knowledge, this is the first application of rectangular nets beyond their study as geometric objects.

\section{Metric Election Games}
\label{sec:models}

We first place several metric election problems in a common framework.
These relationships allow us to derive novel bounds for multiple types of elections under different norms.
For this section, fix a finite-dimensional real vector space $\calX$ equipped with a metric $d$ induced by a norm, a nonempty finite set of voters $V \subseteq \calX$, and a nonempty finite set of candidates $C \subseteq \calX$.

We describe the problems as two-player zero-sum Stackelberg games.
The first player, the \emph{defender}, acts as the leader and chooses a nonempty coalition $S \subseteq \calX$ of at most $k$ candidates.
After observing $S$, the second player, the \emph{challenger}, acts as the follower and chooses a single candidate $c \in \calX$ with the goal of capturing as many voters as possible from $S$.
The candidate $c$ captures a voter if the voter prefers $c$ to every member of $S$.
We take the fraction of voters captured as the challenger's payoff and its negative as the defender's payoff.
Thus, the challenger chooses a best response to $S$, while the defender chooses $S$ anticipating this response.

We consider three types of players, identified by the set of points from which they may choose.
A free player, denoted by $\pF$, may choose any point of $\calX$; a restricted player, denoted by $\pR$, must choose from the predefined set $C$; and a peer-restricted player, denoted by $\pP$, must choose from the set of voters $V$.
Formally, define $\calX_{\pF} \coloneqq \calX$, $\calX_{\pR} \coloneqq C$, and $\calX_{\pP} \coloneqq V$.
The first letter refers to the type of the defender and the second to the type of the challenger.
For $A, B \in \set{\pF, \pR, \pP}$, we denote the largest fraction of voters that a challenger of type $B$ can capture from a defender of type $A$ by

\begin{equation}
    \label{eqn:game-factor}
    {AB}(k) \coloneqq
    \sup_{\substack{V, C \subseteq \calX\\ V, C \text{ finite}}}
    \ \inf_{\substack{S \subseteq \calX_A\\ \card S \leq k}}
    \ \sup_{c \in \calX_B} \frac{\card{\setst{v \in V}{c \succ_v S}}}{\card V}.
\end{equation}
When neither player uses $C$, the supremum over $C$ is irrelevant.
Moreover, for every fixed instance, the defender has an optimal coalition and the challenger has a best response to each coalition, since the payoff takes only finitely many values.
Note that the notation above is the reverse of the convention used in other work on Condorcet-winning sets. Smaller values are better for the first player.
In particular, ${AB}(k) \leq 1 / 2$ means that, in every instance, the defender can choose a coalition of at most $k$ candidates such that no challenger captures more than half of the voters.
This orientation, in which smaller factors give stronger guarantees, matches the usual convention in the $\varepsilon$-net literature.

Considering all pairs of player types gives nine different games. 
In the remainder of this section, we identify the interesting variants and show how this framework recovers multiple problems from the literature.
We build a hierarchy of these variants and then use it to derive improved bounds in the next section.

We begin with a free defender.
When the challenger is also free, $\pF\pF(k)$ is the problem usually referred to as the one-round \emph{Voronoi game}~\cite{BanikJCMS16}.
Restricting the challenger to the predefined set $C$ gives $\pF\pR(k)$, while restricting it to the set of voters $V$ gives $\pF\pP(k)$.
Although either restriction helps the defender, neither improves the worst-case value, as the following proposition shows.

\begin{proposition}
    \label{prp:voronoi-equivalence}
    For every norm and every positive integer $k$, $\pF\pF(k) = \pF\pR(k) = \pF\pP(k)$.
\end{proposition}

Thus, we refer to the three variants collectively as the \emph{Voronoi game}.
We defer the proof of \cref{prp:voronoi-equivalence} to Appendix~\ref{apx:models}.
The idea is to approximate an instance with a free challenger by a sequence of peer-restricted instances constructed by adding increasingly dense challenger locations to the electorate and replacing each original voter by a much larger nearby cluster.
By continuity and strict preferences, the added locations approximate the free challenger's choices, while the clusters preserve the original vote shares and make the added challengers asymptotically negligible.

We next restrict the defender to the predefined candidate set $C$.
When the challenger must also choose from $C$, we obtain $\pR\pR(k)$, the usual problem of finding a Condorcet-winning set in a general metric election~\cite{Elkind15}.
No meaningful bounds can be proved for the other two variants, $\pR\pF(k)$ and $\pR\pP(k)$, since in the worst case the predefined candidates may be arbitrarily far from the set of voters, so that every voter prefers the challenger to the defending coalition.

Lastly, consider the case in which the defender is restricted to the voters.
When the challenger faces the same restriction, we obtain $\pP\pP(k)$, which is known as the \emph{peer selection} problem~\cite{AlonFPT11,AzizLMRW16}.
It is a special case of the Condorcet-winning-set problem obtained by taking $C = V$.
Allowing the defender to choose freely while restricting the challenger to the set of voters can only make the game easier for the defender. 
Hence, $\pP\pP(k) \geq \pF\pP(k)$.
Thus, by \cref{prp:voronoi-equivalence}, peer selection is at least as hard as all Voronoi game variants. Consequently,
\begin{equation}
    \label{eqn:hierarchy_condorcet}
    \pR\pR(k) \geq \pP\pP(k) \geq \pF\pF(k) = \pF\pR(k) = \pF\pP(k).
\end{equation}

Allowing the challenger to be free while the defender is peer-restricted gives $\pP\pF(k)$.
Interestingly, this variant has the same factor as the variant in which the challenger is restricted to a predefined set $C$, namely $\pP\pR(k)$.
To see this, consider all coalitions the peer-restricted defender may choose, and select the best free challenger for each one.
There are only finitely many such coalitions because $V$ is finite, so the set of challengers is finite. Taking $C$ as this set proves the following proposition.

\begin{proposition}
    For every norm and every positive integer $k$, $\pP\pF(k) = \pP\pR(k)$.
\end{proposition}

Thus, we refer to both variants collectively as \emph{strong peer selection}.
Unlike $\pR\pF(k)$ and $\pR\pP(k)$, strong peer selection admits meaningful bounds.
Moreover, it generalizes peer selection, as its name suggests. 
Thus, using \eqref{eqn:hierarchy_condorcet}, we obtain the following hierarchy of problems:
\begin{equation}
    \label{eqn:hierarchy_strong}
    \pP\pF(k) = \pP\pR(k) \geq \pP\pP(k) \geq \pF\pF(k) = \pF\pR(k) = \pF\pP(k).
\end{equation}

Although both general metric elections and strong peer selection are at least as hard as peer selection, they are generally incomparable because $C$ and $V$ need not contain one another.
Restricting to peer selection yields improved bounds for other metric election problems, such as the metric distortion problem~\cite{GkatzelisHS20}.
However, for the Condorcet-winning-set problem, it remains open whether peer selection is any easier than the general case.
We show that in some settings the factors coincide.
In particular, we show in \cref{prop:1_euclidean_equiv} that $\pR\pR_{\ell_2}(1) = \pP\pP_{\ell_2}(1) = \pF\pF_{\ell_2}(1)$.
Likewise, we show in \cref{prop:1_manhattan_equiv} that $\pP\pF_{\ell_1}(1) = \pR\pR_{\ell_1}(1) = \pP\pP_{\ell_1}(1)$.

We summarize the resulting hierarchy of problems in \cref{fig:hierarchy}.
The reductions from weak and strong $\varepsilon$-nets will be the subject of the next section, where we also derive improved bounds for all the described problems.

\begin{figure}[ht!]
    \centering
    \input{figures/hierarchy}

    \caption{Hierarchy of metric election problems and their connections to $\varepsilon$-nets.
    An arrow from $A$ to $B$ means that the factor at $B$ is at most the factor at $A$.
    Solid arrows are unconditional, whereas dashed arrows hold under the metric assumptions indicated beside them.
    Here, $d$ is the metric induced by an arbitrary norm on the plane, and $r$ is an arbitrary positive integer.}
    \label{fig:hierarchy}
\end{figure}
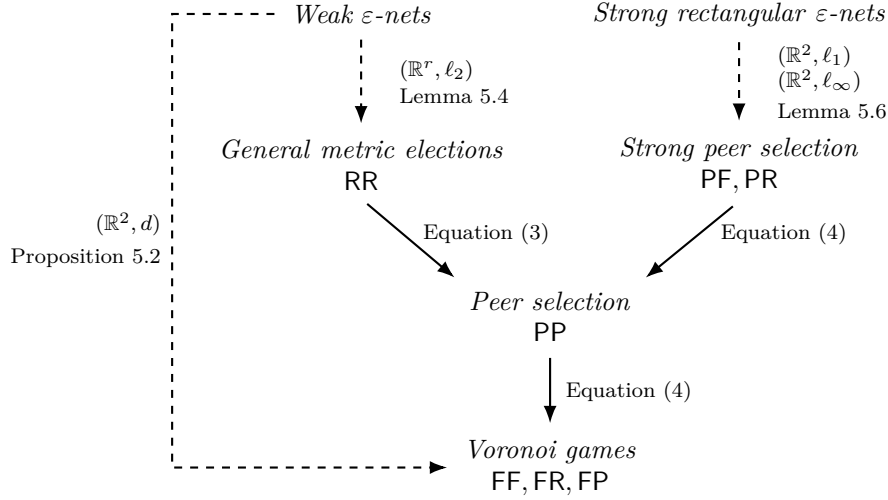

\section{\texorpdfstring{$\bm{\varepsilon}$-nets}{epsilon-nets} and Metric Condorcet Problems}
\label{sec:applications}

The space and the norm determine the geometry of a challenger's dominance region.
A net that intersects every sufficiently large member of a family containing these regions bounds the challenger's payoff.
For Voronoi games, where the defender may choose arbitrary points of the space, a weak $\varepsilon$-net is itself a feasible defending coalition.
This gives a direct connection between this continuous version of the problem and weak convex nets.
The main obstacle to extending this connection to a restricted defender is that a weak net may not be a feasible strategy.
In this section, we extend the Voronoi-game connection to every norm on the plane and show how to overcome these feasibility issues for restricted defenders.

All results rely on the following lemma.
It states that, on the plane, the voters that prefer one point to another can be enclosed in a convex set that still excludes the latter point.

\begin{lemma}
    \label{lem:giveaway}
    On the plane, for any norm and any distinct points $a, b$, we have $b \notin \conv\paren{\Dom(a, b)}$.
    In particular, under the $\ell_1$ norm, $\Dom_{\ell_1}(a, b)$ is contained in an open axis-parallel half-plane that does not contain $b$.
\end{lemma}

We defer the proof of \cref{lem:giveaway} to Appendix~\ref{apx:applications}.
To prove the result, we show how to construct a half-plane containing $\conv\paren{\Dom(a, b)}$ but not $b$.
We also show that the result need not extend to higher dimensions.
In fact, it already fails for every $\ell_p$ norm with $p \neq 2$ on $\R^3$.
For the $\ell_2$ norm, however, the result holds in every dimension because the dominance regions are already convex.

\clearpage
\noindent
\textbf{Weak Nets and the Voronoi Game under Any Norm.}
\citet{BanikJCMS16} observed that the bounds for weak $\varepsilon$-nets extend to Voronoi games under the $\ell_2$ norm.
We first use \cref{lem:giveaway} to extend the reduction to all norms on the plane.

\begin{proposition}
    \label{prop:weak_nets_to_voronoi}
    On the plane, for every norm and every positive integer $k$, $\pF\pF(k) = \pF\pR(k) = \pF\pP(k) \leq \varepsilon_k$.
\end{proposition}

\begin{proof}
    Fix a defending coalition $S \subseteq \R^2$ with $\card S \leq k$ and a challenger $c \in \R^2 \setminus S$.
    Let $W \coloneqq \setst{v \in V}{c \succ_v S}$ be the set of voters captured by $c$.
    For every $s \in S$, we have $W \subseteq \Dom(c, s)$ and therefore
    $\conv\paren{W} \subseteq \conv\paren{\Dom(c, s)}$.
    By \cref{lem:giveaway}, $s \notin \conv\paren{\Dom(c, s)}$, and hence
    $s \notin \conv\paren{W}$.
    Since this holds for every $s \in S$, we obtain
    $\conv\paren{W} \cap S = \varnothing$.
    Hence, the voters captured by $c$ are contained in a convex set disjoint from $S$.
    Thus,
    \begin{equation*}
        \sup_{c \in \R^2} \frac{\card{\setst{v \in V}{c \succ_v S}}}{\card V}
        \leq \sup_{\substack{K \subseteq \R^2 \text{ convex}\\K \cap S = \varnothing}} \frac{\card{K \cap V}}{\card V}.
    \end{equation*}
    Taking the infimum over $S$ and then the supremum over $V$, as in \eqref{eqn:weak-net-factor} and \eqref{eqn:game-factor}, gives $\pF\pF(k) \leq \varepsilon_k$.
    The result follows by \cref{prp:voronoi-equivalence}.
\end{proof}
In particular, \citet[Cor.~2.5]{BergW23} proved that $\varepsilon_4 \leq 1 / 2$ on the plane, which immediately yields the following corollary.
\begin{cor}
    \label{cor:voronoi_general_norm}
    On the plane under any norm, four candidates suffice in the Voronoi game to ensure that no challenger captures a majority of voters; that is, $\pF\pF(4) \leq 1 / 2$.
\end{cor}

\medskip
\noindent
\textbf{Weak Nets and Euclidean Metric Elections.}
In general, the points of a weak net need not belong to the candidate set.
Under the $\ell_2$ norm, this can be repaired by rounding each point of the net to the closest candidate, as the following lemma shows.

\begin{lemma}
    \label{lem:euclidean_projection}
    For all positive integers $r$ and $k$, $\pR\pR_{\ell_2}(k) \leq \varepsilon_k$ on $\R^r$.
\end{lemma}
\begin{proof}
    Fix finite sets of voters $V \subseteq \R^r$ and candidates $C \subseteq \R^r$, and let $Q \subseteq \R^r$ be an $\varepsilon$-net for $V$ with $\card Q \leq k$ and some $\varepsilon > 0$.

    If $Q = \varnothing$, the conclusion is trivial, so we assume $Q \neq \varnothing$.
    For each $q \in Q$, choose a closest candidate $s\paren{q} \in C$ such that $\norm{q - s\paren{q}}_2 = \min_{a \in C} \norm{q - a}_2$, and take $S \coloneqq \setst{s\paren{q}}{q \in Q} \subseteq C$ as the defending coalition.
    For each $c \in C \setminus S$, let the region in which $c$ is preferred to the entire coalition be
    \begin{equation*}
        K_c \coloneqq \bigcap_{s \in S} \Dom_{\ell_2}(c, s).
    \end{equation*}
    Since $\Dom_{\ell_2}(c, s)$ is convex for each $s \in S$, it follows that $K_c$ is convex.
    Moreover, $K_c$ and $Q$ are disjoint.
    Suppose, for contradiction, that there exists $q \in Q \cap K_c$.
    Then
    \begin{equation*}
        \norm{q - c}_2 < \min_{s \in S} \norm{q - s}_2 \leq \norm{q - s\paren{q}}_2,
    \end{equation*}
    a contradiction by the definition of $s\paren{q}$.
    Since $K_c \cap V$ is the set of voters that prefer $c$ to $S$, we obtain
    \begin{equation*}
        \frac{\card{\setst{v \in V}{c \succ_v S}}}{\card V} \leq
        \sup_{\substack{K \subseteq \R^r \text{ convex} \\ K \cap Q = \varnothing}}\frac{\card{K \cap V}}{\card V}.
    \end{equation*}
    Taking the infimum over $Q$ and then the supremum over $V$ and $C$ gives $\pR\pR_{\ell_2}(k) \leq \varepsilon_k$.
\end{proof}

Thus, every upper bound for weak convex nets transfers without loss to general metric elections under the $\ell_2$ norm.
Together with \cite[Cor.~2.5]{BergW23}, this gives the following theorem.

\begin{theorem}
    \label{thm:4_suffice_euclidean}
    Every planar metric election under the $\ell_2$ norm admits a Condorcet-winning set of size at most four, that is, $\pR\pR_{\ell_2}(4) \leq 1 / 2$.
\end{theorem}
To the best of our knowledge, this is the first bound for the $\ell_2$ norm that improves on the general bound of five due to \citet{SongNL26}.
The reduction itself is not restricted to the plane.
The dimensional restriction in \cref{thm:4_suffice_euclidean} comes only from the currently available bounds for weak convex nets.
The nearest-candidate rounding argument, however, is specific to the $\ell_2$ norm.
In Appendix~\ref{apx:applications}, we give an example on the plane that shows that the rounding argument fails for every other $\ell_p$ norm.

\medskip
\noindent
\textbf{Strong Rectangular Nets and Manhattan Strong Peer Selection.}
Under the $\ell_1$ norm, the second part of \cref{lem:giveaway} gives more than convex separation.
Taking the smallest axis-parallel rectangle containing all captured voters produces a rectangle disjoint from the entire defending coalition.
This gives a reduction from strong rectangular nets to strong peer selection.

\begin{lemma}
    \label{lem:strong_nets_to_peer}
    On the plane, for every positive integer $k$, $\pP\pF_{\ell_1}(k) \leq \varsigma_k^{\calR}$.
\end{lemma}
\begin{proof}
    Fix a defending coalition $S \subseteq V$ with $\card S \leq k$ and a challenger $c \in \R^2$ that does not belong to $S$.
    Let $W \coloneqq \setst{v \in V}{c \succ_v S}$ be the set of voters captured by $c$.
    If $W$ is empty, there is nothing to prove.

    Let $R_W$ be the smallest closed axis-parallel rectangle containing $W$. 
    For every $s \in S$, we have $W \subseteq \Dom_{\ell_1}(c, s)$.
    By \cref{lem:giveaway}, $W$ is contained in an open axis-parallel half-plane that does not contain $s$.
    Since $W$ is finite, it lies at positive distance from the boundary of this half-plane.
    Because the half-plane is axis-parallel, $R_W$ is also strictly contained in it.
    Hence, $s \notin R_W$.
    As this holds for every $s \in S$, we obtain $R_W \cap S = \varnothing$.

    Thus, the voters captured by $c$ lie in an axis-parallel rectangle disjoint from $S$.
    It follows that
    \begin{equation*}
        \sup_{c \in \R^2} \frac{\card{\setst{v \in V}{c \succ_v S}}}{\card V}
        \leq \sup_{\substack{R \in \calR\\R \cap S = \varnothing}} \frac{\card{R \cap V}}{\card V}.
    \end{equation*}
    Taking the infimum over $S$ and then the supremum over $V$ gives $\pP\pF_{\ell_1}(k) \leq \varsigma_k^{\calR}$.
\end{proof}

The same bound holds under the $\ell_\infty$ norm, since rotating the plane by $\pi / 4$ transforms $\ell_\infty$ distances into $\ell_1$ distances up to a multiplicative factor of $\sqrt{2}$.
Hence, $\pP\pF_{\ell_\infty}(k) \leq \varsigma_k^{\calR}$ for every positive integer $k$.
Thus, every upper bound for strong rectangular nets transfers without loss to strong peer selection under the $\ell_1$ and $\ell_\infty$ norms.
In particular, \cref{thm:strong_R_3_nets} yields the following theorem.

\begin{theorem}
    \label{thm:3_suffice_strong_peer}
    Three voters suffice for strong peer selection on the plane under the $\ell_1$ and $\ell_\infty$ norms, that is, $\pP\pF_{\ell_1}(3) \leq 1 / 2$ and $\pP\pF_{\ell_\infty}(3) \leq 1 / 2$.
\end{theorem}

Consequently, three voters also suffice for peer selection under both norms, yielding a Condorcet-winning set of size at most three in the special case $C = V$.

\medskip
\noindent
\textbf{The One-Candidate Case.}
Finally, we consider the case in which the defending coalition consists of a single candidate.
In this case, we are able to determine the exact factors for the problems considered in this work.
We begin with the $\ell_2$ norm.

\begin{proposition}
    \label{prop:1_euclidean_equiv}
    On the plane, $\pR\pR_{\ell_2}(1) = \pP\pP_{\ell_2}(1) = \pF\pF_{\ell_2}(1) = 2 / 3$.
\end{proposition}
\begin{proof}
    \citet[Thm.~1]{BanikJCMS16} proved that $\varepsilon_1 \leq 2 / 3$.
    Applying \cref{lem:euclidean_projection} gives $\pR\pR_{\ell_2}(1) \leq \varepsilon_1 \leq 2 / 3$.

    Now, we show that $\pF\pF_{\ell_2}(1) \geq 2 / 3$.
    Let $V$ be the vertex set of a nondegenerate triangle in the plane and fix any point $p \in \R^2$ chosen by the free defender.
    There is a half-plane $H$ that contains two vertices and satisfies $p \notin H$.
    Let $c$ be the reflection of $p$ across the boundary of $H$.
    Then $H = \Dom_{\ell_2}(c, p)$, so $c$ captures at least two of the three voters and $\pF\pF_{\ell_2}(1) \geq 2 / 3$.

    Therefore, by \eqref{eqn:hierarchy_condorcet}, $2 / 3 \geq \pR\pR_{\ell_2}(1) \geq \pP\pP_{\ell_2}(1) \geq \pF\pF_{\ell_2}(1) \geq 2 / 3$, which implies equality.
\end{proof}

The same conclusion does not hold for strong peer selection, since the problem has no nontrivial guarantee under the $\ell_2$ norm with a single candidate.
To see this, place $n$ voters at the vertices of a convex polygon.
For each voter $v$, choose a line that strictly separates $v$ from the convex hull of the remaining voters, and let $c$ be the reflection of $v$ across this line.
The open half-plane containing the remaining voters is then $\Dom_{\ell_2}(c, v)$, so the challenger captures all $n - 1$ other voters.
Letting $n$ tend to infinity gives $\pP\pF_{\ell_2}(1) = \pP\pR_{\ell_2}(1) = 1$.

We next turn to the $\ell_1$ and $\ell_\infty$ norms.
Here, it is the Voronoi game that behaves differently from the other models.
\begin{proposition}
    \label{prop:1_manhattan_equiv}
    On the plane, $\pP\pF_{\ell_1}(1) = \pR\pR_{\ell_1}(1) = \pP\pP_{\ell_1}(1) = 3 / 4$.
    The same equality holds under the $\ell_\infty$ norm.
\end{proposition}
\begin{proof}
    We prove the result for the $\ell_\infty$ norm. First, \citet[Thm.~2]{AshokAG14} proved that $\varsigma_1^\calR = 3 / 4$.
    By \cref{lem:strong_nets_to_peer} and the rotation argument above, $\pP\pF_{\ell_\infty}(1) \leq \varsigma_1^\calR = 3 / 4$.

    Second, in any metric election under the $\ell_\infty$ norm, at most four candidates can appear last in the voters' rankings, each attaining an extreme horizontal or vertical coordinate.
    Suppose that every candidate is defeated by some challenger by more than $3 / 4$ of the voters.
    Choosing one such challenger for every candidate defines a directed graph in which every vertex has out-degree one.
    This graph contains a directed cycle $c_0, c_1, \ldots, c_{m - 1}$ such that $c_{i + 1}$ defeats $c_i$ by more than $3 / 4$ of the voters, where the indices are taken modulo $m$.
    Define
    \begin{equation*}
        V_{i} \coloneqq \setst{v \in V}{c_i \succ_v c_{i+1}}.
    \end{equation*}
    Then, since $c_{i + 1}$ defeats $c_i$ by more than $3 / 4$ of the voters, we have $\card{V_{i}} < \card{V} / 4$ for $i = 0, \dots, m - 1$.
    Now, restrict each voter's ranking to the candidates in the cycle.
    For each candidate $c_i$ that appears last in at least one restricted ranking, every voter who ranks $c_i$ last belongs to $V_{i - 1}$.
    Consequently, the sets $V_{i - 1}$ corresponding to these candidates cover $V$.
    There are at most four such sets, and each contains less than a quarter of the voters, so their union cannot cover $V$.
    This contradiction shows that $\pR\pR_{\ell_\infty}(1) \leq 3 / 4$.

    Now, we show that $\pP\pP_{\ell_\infty}(1) \geq 3 / 4$.
    We first construct a metric election under the $\ell_\infty$ norm in which all voter rankings are cyclic permutations of one another.
    We then convert it into a peer-selection instance by replacing each voter with a large cluster of nearby voters and including the candidates in the voter set.
    Let $R\paren{x, y} \coloneqq \paren{-y, x}$ denote rotation by $\pi / 2$, and take all indices modulo four.
    Set $c_1 \coloneqq \paren{2, 1}$ and $v_1 \coloneqq \paren{3, 4}$, and define $c_i \coloneqq R\paren{c_{i - 1}}$ and $v_i \coloneqq R\paren{v_{i - 1}}$ for $i = 2, 3, 4$ (see \cref{fig:manhattan_one_candidate}).
    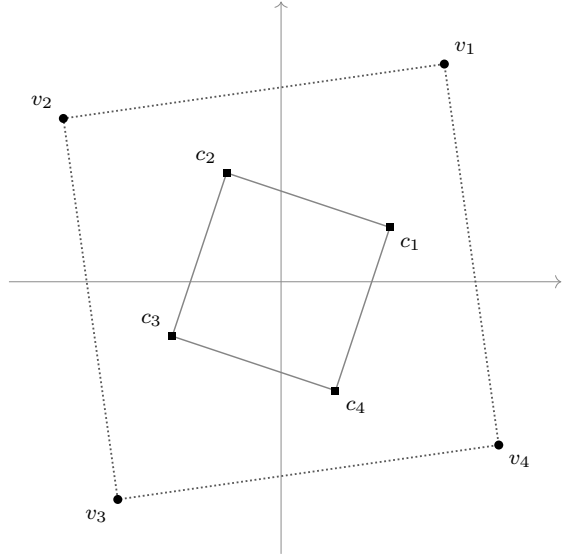
\begin{figure}[ht!]
        \centering
        \input{figures/manhattan-one-candidate}
        \caption{Lower bound used in the proof of \cref{prop:1_manhattan_equiv}.}
        \label{fig:manhattan_one_candidate}
    \end{figure}
    The distances from $v_1$ to $c_1, c_2, c_3, c_4$ are $3, 4, 5, 6$, respectively.
    By symmetry, the ranking at each $v_i$ is
    \begin{equation}
        \label{eqn:1_manhattan_equiv_eq_1}
        c_i \succ_{v_i} c_{i + 1} \succ_{v_i} c_{i + 2} \succ_{v_i} c_{i + 3}.
    \end{equation}
    Consequently, for every $i \in [4]$, $c_{i - 1}$ is preferred to $c_i$ at each $v_j$ with $j \neq i$.
    Moreover, every voter is closer to every candidate than to any other voter,
    \begin{equation}
        \label{eqn:1_manhattan_equiv_eq_2}
        \max_{i, j \in [4]} \norm{v_j - c_i}_\infty < \min_{\substack{i, j \in [4] \\ i \neq j}} \norm{v_j - v_i}_\infty.
    \end{equation}

    Fix a positive integer $q$.
    For every $i \in [4]$, replace $v_i$ by a cluster $U_{i, q}$ of $q$ distinct points sufficiently close to $v_i$.
    Choose the clusters sufficiently small that the rankings in \eqref{eqn:1_manhattan_equiv_eq_1} are preserved and, whenever $i \neq j$, every voter in $U_{j, q}$ is closer to $c_i$ than to every point in $U_{i, q}$.
    Take $V_q \coloneqq \set{c_1, c_2, c_3, c_4} \cup \bigcup_{i = 1}^4 U_{i, q}$ as a peer selection instance.
    Suppose first that the defender chooses $c_i$ for some $i \in [4]$.
    Then $c_{i - 1}$ is preferred to $c_i$ by every voter in $U_{j, q}$ with $j \neq i$, and therefore the challenger captures at least $3q$ voters.

    Suppose instead that the defender chooses $s \in U_{i, q}$ and the challenger chooses $c_i$.
    By \eqref{eqn:1_manhattan_equiv_eq_2}, every voter in $U_{j, q}$ with $j \neq i$ is closer to $c_i$ than to $s$.
    Thus, $c_i$ again captures at least $3q$ voters.

    Thus, for every choice of the defender, the challenger captures at least $3q$ of the $4q + 4$ voters.
    Letting $q \to \infty$ shows that $\pP\pP_{\ell_\infty}(1) \geq 3 / 4$.
    Since $\pP\pF_{\ell_\infty}(1) \leq 3 / 4$, \eqref{eqn:hierarchy_strong} yields $3 / 4 \geq \pP\pF_{\ell_\infty}(1) \geq \pP\pP_{\ell_\infty}(1) \geq 3 / 4$.
    The upper bound for $\pR\pR_{\ell_\infty}(1)$ established above and \eqref{eqn:hierarchy_condorcet} similarly yield $3 / 4 \geq \pR\pR_{\ell_\infty}(1) \geq \pP\pP_{\ell_\infty}(1) \geq 3 / 4$.
    Hence, equality holds throughout.
    Rotating the plane by $\pi / 4$ and scaling by $1 / \sqrt{2}$ gives the same equalities under the $\ell_1$ norm.
\end{proof}

We now compare these values with the Voronoi game, whose one-candidate factor is smaller under the $\ell_1$ and $\ell_\infty$ norms.
Consider the $\ell_1$ norm and let the defender choose a coordinate-wise median $p$ of the voter set.
By \cref{lem:giveaway}, the voters captured by any challenger are contained in an open axis-parallel half-plane that does not contain $p$.
By the choice of $p$, every such half-plane contains at most half of the voters. Thus, $\pF\pF_{\ell_1}(1) \leq 1 / 2$.
An instance consisting of two distinct voters gives the reverse inequality.
Therefore, $\pF\pF_{\ell_1}(1) = 1 / 2$.
By the same rotation argument, $\pF\pF_{\ell_\infty}(1) = 1 / 2$.

\section{A Norm-Independent Two-Candidate Bound}
\label{sec:general_norms}

We finally turn to a guarantee that is independent of the choice of the norm.
Consider an election in which all candidates lie at the same distance from some voter.
Geometrically, the candidates may be placed arbitrarily on the boundary of a unit ball centered at the origin. From the perspective of the election, this corresponds to a voter who is indifferent among all candidates.
This simple structure is enough to guarantee a Condorcet-winning pair under every norm on the plane.
Moreover, the pair can be required to contain any prescribed candidate.
Throughout this section, we drop the general-position assumption.
For $x \in \calX$ and $\alpha \geq 0$, we denote the closed ball of radius $\alpha$ centered at $x$ by $B\paren{x, \alpha}$.

\begin{proposition}
    \label{prop:undecided_voter}
    On the plane, under every norm, if an instance contains an undecided voter that is equidistant from every candidate, then every candidate belongs to a Condorcet-winning set of size at most two.
\end{proposition}

\begin{proof}
    Fix a set of voters $V = \curly{v_1, \dots, v_n}$, an undecided voter $u \in V$, and a candidate $c_1 \in C$.
    We may assume that $u$ is at the origin and at unit distance from every candidate.
    Thus, every candidate lies on the boundary of the unit ball $B=B(0,1)$.
    For every voter $v \in V$ and candidate $c \in C$, let
    \begin{equation*}
        A_v(c) \coloneqq \setst{x \in \partial B}{d\paren{v, x} \leq d\paren{v, c}}.
    \end{equation*}
    These are the points on $\partial B$ that $v$ weakly prefers to $c$.
    The ball $B\paren{v, d\paren{v, c}}$ is a translated and scaled copy of $B$.
    The boundaries of a planar convex body and any translated and scaled copy of it cross at most twice~\cite[Sec.~2.2]{AgarwalPS08}.
    Hence, $A_v(c)$ is always either a point, a closed arc of $\partial B$, or all of $\partial B$.
    Consequently, the candidates at least as close to $v$ as $c$ occur consecutively around $\partial B$.

    For every voter $v \in V$, choose an arbitrary closest candidate $\operatorname{top}\paren{v} \in C$.
    For a candidate $c \in C$, define its \emph{plurality score} by $\plu\paren{c} \coloneqq \card{\setst{v \in V}{\operatorname{top}\paren{v} = c}}$, and for $X \subseteq C$, define $\plu\paren{X} \coloneqq \sum_{c \in X} \plu\paren{c}$.

    Suppose first that $\plu\paren{c} \geq n / 2$ for some $c \in C$.
    Every voter whose top candidate is $c$ is at least as close to $c$ as to any challenger, so any candidate outside $\set{c_1, c}$ is preferred to this pair by at most $n - \plu\paren{c} \leq n / 2$ voters.
    Hence, $\set{c_1, c}$ is a Condorcet-winning set.
    We may therefore assume that $\plu\paren{c} < n / 2$ for every $c \in C$.
    
    Starting immediately after $c_1$, traverse the candidates in clockwise order around $\partial B$, and let $c_2$ be the first candidate for which the accumulated plurality score, including $\plu\paren{c_2}$, exceeds $n / 2$.
    Such a candidate exists because the candidates other than $c_1$ have total plurality score $n - \plu\paren{c_1} > n / 2$.
    Let $I_1$ be the candidates strictly between $c_1$ and $c_2$ in this clockwise order, and let $I_2$ be the candidates strictly between $c_2$ and $c_1$.
    By the choice of $c_2$, we have $\plu\paren{I_1} \leq n / 2$ and $\plu\paren{I_1} + \plu\paren{c_2} > n / 2$.
    Therefore,
    \begin{equation*}
        \plu\paren{I_2} = n - \plu\paren{c_1} - \plu\paren{c_2} - \plu\paren{I_1} < \frac{n}{2}.
    \end{equation*}

    Take a challenger $b \in I_i$ for $i \in \set{1, 2}$.
    If a voter $v$ prefers $b$ to both $c_1$ and $c_2$, then $A_v(b)$ contains both $b$ and $\operatorname{top}\paren{v}$ but contains neither $c_1$ nor $c_2$.
    Since $A_v(b)$ is a connected arc, it follows that $\operatorname{top}\paren{v} \in I_i$.
    Thus,
    \begin{equation*}
        \card{\setst{v \in V}{b \succ_v \set{c_1, c_2}}} \leq \plu\paren{I_i} < \frac{n}{2}.
    \end{equation*}
    Hence, $\set{c_1, c_2}$ is a Condorcet-winning set.
    Since $c_1$ was arbitrary, the result follows.
\end{proof}

The proof only uses the convexity of the unit ball, so the result also extends to asymmetric gauges.
The proof also does not require the center of the unit ball to be a voter; it only requires all candidates to lie on its boundary.
The example in \cref{fig:undecided_voter_example} shows that this more general statement is tight for every $\ell_p$ norm.
For every $1 \leq p \leq \infty$, the three candidates are equidistant from the origin, while the voters induce a Condorcet cycle, so no candidate is a Condorcet winner.

\begin{figure}[H]
    \centering
    \input{figures/undecided-voter-example}
    \caption{Example producing a Condorcet cycle under every $\ell_p$ norm.
    The candidates are $a = \paren{1, 0}$, $b = \paren{0, 1}$, and $c = \paren{-r_p, -r_p}$, where $r_p \coloneqq 2^{-1 / p}$ for $p < \infty$ and $r_\infty \coloneqq 1$, and the voters are $v_1 = \paren{2, 1 / 4}$, $v_2 = \paren{-2, 7 / 4}$, and $v_3 = \paren{1 / 4, -2}$.
    The segment along the negative diagonal traces the position of $c$ as $p$ increases.}
    \label{fig:undecided_voter_example}
\end{figure}
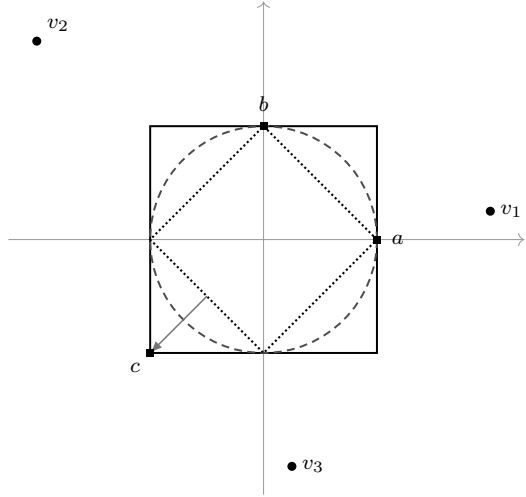

\section{Conclusion}

Representing political preferences in the plane is familiar both from political discussions in the media and from academic work across economics, political science, mathematics, and computer science.
Its simplicity makes planar metric models particularly appealing for studying collective decisions.
In this paper, we use the geometry of these models to obtain sharper guarantees for Condorcet-winning sets and related selection problems.

The notion of Condorcet-winning sets introduced by \citet{Elkind15} provides a natural way to address Condorcet's paradox and may also be useful in practical selection problems.
In many such problems, the alternatives are the participants themselves, making peer selection the relevant restriction.
We place the Condorcet problem, peer selection, strong peer selection, and the Voronoi game in a common two-player framework that separates the restrictions imposed on the defending coalition and the challenger.
This framework makes the relationships among these problems explicit and, more importantly, connects them to weak and strong $\varepsilon$-nets.

These connections yield improved bounds under the main planar norms.
Under the $\ell_2$ norm, rounding weak $\varepsilon$-nets to the candidate set shows that four candidates suffice for the general Condorcet-winning-set problem, improving the general guarantee of five candidates~\cite{SongNL26}.
Under the $\ell_1$ and $\ell_\infty$ norms, our bound $\varsigma_3^{\calR} \leq 1 / 2$ for strong rectangular nets shows that three voters suffice even for strong peer selection, and hence also for peer selection.
The geometric result improves the previous upper bound of $9 / 16$~\cite{AshokAG14}, while the lower bound for two points shows that three is the minimum size that can guarantee the threshold $1 / 2$.
For the small collective decisions that motivate peer selection, a three-person committee is also a natural and practically meaningful guarantee.

Several questions remain open.
The most immediate is whether peer selection is easier than the Condorcet-winning-set problem for general metric elections.
Even when the defending coalition consists of a single candidate, it remains open whether $\pR\pR_{\ell_p}(1) = \pP\pP_{\ell_p}(1)$ for $1 < p < \infty$ and $p \neq 2$.

On the geometric side, we believe that \cref{thm:strong_R_3_nets} is tight, that is, that $\varsigma_3^{\calR} = 1 / 2$.
However, it seems that a matching lower-bound construction will require different ideas.
More broadly, to the best of our knowledge, this is the first use of $\varepsilon$-nets to derive bounds for these restricted Condorcet problems, and the connection provides a concrete route toward further improvements.
For weak nets, the current bounds for three points are $5 / 11 \leq \varepsilon_3 \leq 8 / 15$~\cite{MustafaR09}.
An improvement to $\varepsilon_3 \leq 1 / 2$ would, through \cref{lem:euclidean_projection}, show that three candidates suffice under the $\ell_2$ norm.

\section*{Acknowledgements}

We thank Cristian Palma Foster for his contributions during the early stages of this work.
We thank Claire Chang, Frank Connor, and Haripriya Pulyassary for helpful discussions.
We thank Paul Gölz, Omar El Housni, and David Shmoys for their guidance and support.
Finally, we thank the participants in ORIE 6150, whose discussions motivated us to study this problem.

\section*{AI Disclosure}
ChatGPT was used during the drafting of this manuscript to generate figures, revise the writing, and verify the proofs.
For the technical work, ChatGPT generated code to search for lower bounds on strong rectangular $\varepsilon$-nets and to test, using mixed-integer programming, the coverage of a point configuration and all its symmetries in the proof of the upper bounds.

\bibliographystyle{plainnat}
\bibliography{bibliography}

\clearpage
\appendix

\section{Deferred details for \nameref{sec:snets}}
\label{apx:snets}

\medskip
\noindent

\textbf{Proof of \cref{thm:strong_R_3_nets}.} We first introduce some notation that will be used throughout the proof.
Fix a finite point set $V \subseteq \R^2$. For any set $X \subseteq \R^2$, write
\begin{equation*}
    \mu\paren{X} \coloneqq \frac{\card{X \cap V}}{\card V}.
\end{equation*}
We use the quartile grid and cell labels from \cref{fig:grid}, and let $p$ denote the intersection of its two middle lines.
We number the rows from north to south and the columns from west to east.
We call a rectangle \emph{dangerous} if its mass is greater than $1 / 2$, and \emph{safe} otherwise.
Every rectangle that does not contain $p$ is contained in two rows or two columns and is therefore safe.

For a point $z$, define the regions
\begin{align*}
    \mathsf W(z) &\coloneqq \setst{q \in \R^2}{x\paren{q} \leq x(z)}, &
    \mathsf E(z) &\coloneqq \setst{q \in \R^2}{x\paren{q} \geq x(z)}, \\
    \mathsf N(z) &\coloneqq \setst{q \in \R^2}{y\paren{q} \geq y(z)}, &
    \mathsf S(z) &\coloneqq \setst{q \in \R^2}{y\paren{q} \leq y(z)}
\end{align*}
where $x$ and $y$ denote the coordinate maps.
We denote their respective interiors by $\Wof{z}$, $\Eof{z}$, $\Nof{z}$, and $\Sof{z}$. 
We say that a rectangle $R \in \calR$ is \emph{east of $z$} if it is contained in $\Eof{z}$, and similarly for the other three directions.

We use two types of selected points. The first type consists of \emph{extreme points}. For a nonempty subset $X \subseteq V$, the symbols
\begin{equation*}
    \West\paren{X}, \qquad \East\paren{X}, \qquad \North\paren{X}, \qquad \South\paren{X}
\end{equation*}
denote, respectively, a westernmost, easternmost, northernmost, and southernmost point of $X$. Thus, $X \cap \Wof{\West\paren{X}} = \varnothing$.

The second type consists of \emph{split points}, which control the mass in two directions simultaneously. We define split points using the following generalization of \citet[Lem.~3]{AshokAG14}.

\begin{lemma}
    \label{lem:split-point}
    Let $X \subseteq V$ be nonempty, and let $a, b \geq 0$ satisfy $a + b \leq \mu\paren{X}$. 
    There is a point $z \in X$ such that
    \begin{equation*}
        \mu\paren{X \cap \Eof{z}} \leq \mu\paren{X} - a
        \qquad\text{and}\qquad
        \mu\paren{X \cap \Sof{z}} \leq \mu\paren{X} - b.
    \end{equation*}
    We call such a point a \emph{southeast $\paren{a, b}$-split point} of $X$ and denote one such point by $p_X^{\mathrm{SE}}\paren{a, b}$. The other three orientations are defined analogously.
\end{lemma}

\begin{proof}
    Let $m \coloneqq \card X$, $n \coloneqq \card V$, and set $k \coloneqq \max\paren{1, \ceiling{bn}}$.
    The assumption $b \leq \mu\paren{X} = m / n$, together with the fact that $X$ is nonempty, ensures that $1 \leq k \leq m$.

    Consider the $k$ easternmost points of $X$ and the $m - k + 1$ southernmost points of $X$, breaking ties arbitrarily. The sum of their cardinalities is $m + 1$, whereas $X$ contains only $m$ points. They must therefore intersect. Let $z$ be a point in their intersection.

    Because $z$ is one of the $m - k + 1$ southernmost points of $X$, at most $m - k$ points of $X$ lie strictly south of $z$. 
    Thus, at least $k$ points of $X$ lie weakly north of $z$, and hence
    \begin{equation*}
        \mu\paren{X \cap \mathsf N(z)}
        \geq \frac{k}{n}
        \geq b.
    \end{equation*}
    Similarly, because $z$ is one of the $k$ easternmost points of $X$, at most $k - 1$ points of $X$ lie strictly east of $z$. 
    Therefore, at least $m - k + 1$ points of $X$ lie weakly west of $z$. Since $k \leq bn + 1$, we obtain
    \begin{equation*}
        \mu\paren{X \cap \mathsf W(z)}
        \geq \frac{m - k + 1}{n}
        \geq \frac{m - bn}{n}
        = \mu\paren{X} - b
        \geq a.
    \end{equation*}
    Since $\mathsf W(z)$ and $\Eof{z}$ partition the plane, and likewise $\mathsf N(z)$ and $\Sof{z}$, the desired inequalities follow.
\end{proof}

Thus, a rectangle avoiding the point $z$ given by \cref{lem:split-point} contains at most $\mu\paren{X} - a$ mass from $X$ if it lies east of $z$, and at most $\mu\paren{X} - b$ if it lies south of $z$. When $a = b = \mu\paren{X} / 2$, we call the point \emph{balanced} and omit $\paren{a, b}$. In the illustrations throughout the proof, circular markers denote extreme points, whereas square markers denote split points given by \cref{lem:split-point}. To simplify the notation, we use $A$ and $\mu\paren{A}$ interchangeably (likewise for all other grid cells). We use the following identity repeatedly throughout the proof. It states that the sum of the masses of the center cells equals the sum of the masses of the corner cells:
\begin{equation}
    \label{eqn:corner-center}
    A + D + M + P = F + G + J + K.
\end{equation}
Indeed, the union of the two central rows and the union of the two exterior columns each have mass $1 / 2$; subtracting the mass $E + I + H + L$ common to both yields \eqref{eqn:corner-center}.
We first cover three simple cases.

\begin{lemma}[All four center cells are empty]
    \label{lem:empty-center}
    If $F = G = J = K = 0$, then there is a strong rectangular $1 / 2$-net of size at most three.
\end{lemma}

\begin{proof}
    By \eqref{eqn:corner-center}, $A = D = M = P = 0$. Consequently, each of the following four unions of cells has mass $1 / 4$:
    \begin{equation*}
        N_0 \coloneqq B \cup C,
        \qquad E_0 \coloneqq H \cup L,
        \qquad S_0 \coloneqq N \cup O,
        \qquad W_0 \coloneqq E \cup I.
    \end{equation*}
    Take $q_1 = \South\paren{N_0}$ and $q_2 = \North\paren{S_0}$. Any rectangle $R \in \calR$ that is south of $q_1$ and north of $q_2$ is safe because it misses $N_0$ and $S_0$, whose total mass is $1 / 2$. Suppose, without loss of generality, that $q_1 \in B$. Then,
    \begin{itemize}
        \item if $q_2 \in O$, $Q = \set{q_1, q_2}$ is a strong rectangular $1 / 2$-net. Indeed, a dangerous rectangle $R \in \calR$ would need to be south or east of $q_1$ and north or west of $q_2$. The only possibility is then that $R$ is south of $q_1$ and north of $q_2$, which is safe.
        \item if $q_2 \in N$, any dangerous rectangle avoiding $q_1$ and $q_2$ must lie east of both points. Taking $q_3 = \West\paren{E_0}$ then ensures that $Q = \set{q_1, q_2, q_3}$ is a strong rectangular $1 / 2$-net.
    \end{itemize}
    This concludes the proof.
\end{proof}

\begin{lemma}[All four center cells are nonempty]
    \label{lem:full-center}
    If $F, G, J, K > 0$, then there is a strong rectangular $1 / 2$-net of size at most three.
\end{lemma}

\begin{proof}
    We use one of the following four rotationally symmetric choices:
    \begin{align*}
        q_1 &= \West\paren{G}, & q_2 &= \North\paren{J}, & q_3 &= p_K^{\mathrm{NW}}
        &&\text{if } P \leq G + J + K / 2, \\
        q_1 &= \East\paren{F}, & q_2 &= \North\paren{K}, & q_3 &= p_J^{\mathrm{NE}}
        &&\text{if } M \leq F + K + J / 2, \\
        q_1 &= \South\paren{F}, & q_2 &= \West\paren{K}, & q_3 &= p_G^{\mathrm{SW}}
        &&\text{if } D \leq F + K + G / 2, \\
        q_1 &= \South\paren{G}, & q_2 &= \East\paren{J}, & q_3 &= p_F^{\mathrm{SE}}
        &&\text{if } A \leq G + J + F / 2.
    \end{align*}

    \Cref{fig:full-center} displays, for each choice, the only possible position of a rectangle avoiding the selected points that is not contained in two rows or two columns.

    \begin{figure}[H]
        \centering
        \resizebox{\linewidth}{!}{\input{figures/strong-net-full-center}}
        \caption{The four possible configurations when no center cell is empty.}
        \label{fig:full-center}
    \end{figure}
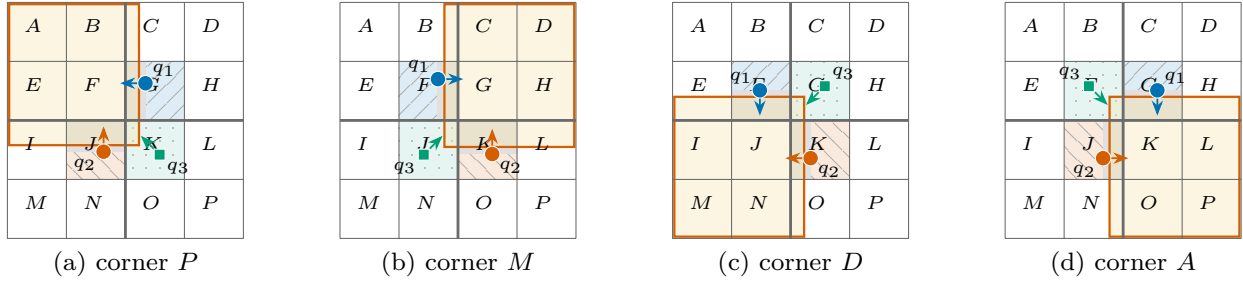

    Consider the first choice. Every rectangle avoiding the selected points, except one in the position shown in \cref{fig:full-center}(a), is contained in two rows or two columns. In the remaining case, the rectangle $R$ is west of $q_1 = \West\paren{G}$ and north of $q_2 = \North\paren{J}$. It is therefore contained in the first three rows and first three columns, a region of mass $1 / 2 + P$. It contains no point of $G$ or $J$, and the balanced point $q_3$ restricts it to at most $K / 2$ mass from $K$. Consequently,
    \begin{equation*}
        \mu\paren{R} \leq 1 / 2 + P - G - J - K / 2 \leq 1 / 2.
    \end{equation*}
    The other three choices are rotations of this argument. Furthermore, at least one of the four displayed conditions holds. Otherwise, summing their negations and using \eqref{eqn:corner-center} gives
    \begin{equation*}
        A + D + M + P > \frac{5}{2}\paren{F + G + J + K} = \frac{5}{2}\paren{A + D + M + P},
    \end{equation*}
    a contradiction because the center is nonempty.
\end{proof}

\begin{lemma}[Two opposite center cells are nonempty]
    \label{lem:opposite-center}
    If exactly two opposite center cells are nonempty, then there is a strong rectangular $1 / 2$-net of size at most three.
\end{lemma}

\begin{proof}
    By symmetry, assume $G = J = 0$ and $F, K > 0$. Choose
    \begin{equation*}
        q_1 = p_F^{\mathrm{SE}},
        \qquad
        q_2 = p_K^{\mathrm{NW}},
    \end{equation*}
    both balanced. A rectangle avoiding $q_1$ and $q_2$ must lie east or south of $q_1$ and west or north of $q_2$. A rectangle east of $q_1$ and west of $q_2$ lies in the middle two columns, while one south of $q_1$ and north of $q_2$ lies in the middle two rows. Thus, both have mass at most $1 / 2$.

    The other two possibilities are northeast and southwest corner rectangles. Their masses satisfy
    \begin{equation}
        \label{eqn:opposite-two-corners}
        \mu\paren{R_{\mathrm{NE}}}
        \leq 1 / 2 + M - \frac{F + K}{2},
        \qquad
        \mu\paren{R_{\mathrm{SW}}}
        \leq 1 / 2 + D - \frac{F + K}{2}.
    \end{equation}
    These two upper bounds cannot both exceed $1 / 2$, since that would imply $D + M > F + K$, whereas \eqref{eqn:corner-center} gives
    \begin{equation*}
        F + K = A + D + M + P.
    \end{equation*}

    If both bounds in \eqref{eqn:opposite-two-corners} are at most $1 / 2$, then $q_1$ and $q_2$ already form the required net. Otherwise, without loss of generality, assume that the bound for $R_{\mathrm{NE}}$ exceeds $1 / 2$, and let $Z \coloneqq C \cup D \cup H$.
    The row and column sums give $C + D + H = I + M + N$.
    Together with \eqref{eqn:corner-center}, this yields
    \begin{equation}
        \label{eqn:opposite-Z}
        Z + F + K - 2M = A + D + P + I + N \geq 0.
    \end{equation}
    Suppose that some northeast rectangle avoiding $q_1, q_2$ has mass strictly larger than $1 / 2$. Then the first bound in \eqref{eqn:opposite-two-corners} must itself be strictly larger than $1 / 2$, so $2M > F + K$. Equation~\eqref{eqn:opposite-Z} then implies $Z > 0$.
    Choose
    \begin{equation*}
        q_3 = p_Z^{\mathrm{SW}},
    \end{equation*}
    a balanced southwest split point of $Z$. Because $Z$ lies northeast of $p$, a rectangle containing $p$ and avoiding $q_3$ is west or south of $q_3$, and hence contains at most $Z / 2$ mass from $Z$. Every northeast rectangle avoiding all three points therefore satisfies
    \begin{equation*}
        \mu\paren{R}
        \leq 1 / 2 + M - \frac{F + K + Z}{2}
        \leq 1 / 2
    \end{equation*}
    by \eqref{eqn:opposite-Z}. \Cref{fig:opposite-center} shows the two ways in which such a rectangle can avoid $q_3$.

\end{proof}

\begin{figure}[H]
    \centering
    \input{figures/strong-net-opposite-center}
    \caption{The only potentially dangerous opposite-center rectangle, shown in the northeast orientation. The third point $q_3$ is the balanced southwest split point of $C \cup D \cup H$.}
    \label{fig:opposite-center}
\end{figure}
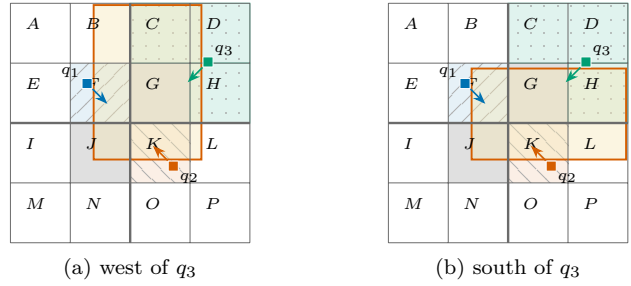

\medskip
\noindent
\textbf{Interlude: A common configuration.} From now on, assume, without loss of generality, that
\begin{equation}
    \label{eqn:normalized-center}
    K = 0 \text{ and }F > 0.
\end{equation}
The row and column sums then give the following refined identities:
\begin{align*}
    A + B + D &= G + O, & A + C + D &= F + J + N, \\
    A + E + M &= J + L, & A + I + M &= F + G + H.
\end{align*}

We first distinguish whether
\begin{equation*}
    A > G + J + F / 2
    \qquad\text{or}\qquad
    A \leq G + J + F / 2.
\end{equation*}
We begin with the first case, which does not depend on whether $G$ or $J$ is empty.

\begin{lemma}[Large northwest corner]
    \label{lem:large-A}
    Under \eqref{eqn:normalized-center}, if
    \begin{equation*}
        A > G + J + F / 2,
    \end{equation*}
    then there is a strong rectangular $1 / 2$-net of size at most three.
\end{lemma}

\begin{proof}
    The refined identities give
    \begin{equation*}
        O = A + B + D - G > J + F / 2,
        \qquad
        L = A + E + M - J > G + F / 2.
    \end{equation*}
    Thus, the sets $X \coloneqq H \cup L$ and $Y \coloneqq N \cup O$ are nonempty. Choose the three balanced split points
    \begin{equation*}
        q_1 = p_F^{\mathrm{SE}},
        \qquad
        q_2 = p_X^{\mathrm{SW}},
        \qquad
        q_3 = p_Y^{\mathrm{NE}}.
    \end{equation*}
    Every rectangle containing $p$ and avoiding these points is either contained in two rows or two columns, or contained in one of the five rectangles $R_1, \dots, R_5$ shown in \cref{fig:outward}.

    \begin{figure}[H]
        \centering
        \resizebox{\linewidth}{!}{\input{figures/strong-net-outward}}
        \caption{The five remaining positions of a rectangle containing $p$ and avoiding the three selected points.}
        \label{fig:outward}
    \end{figure}
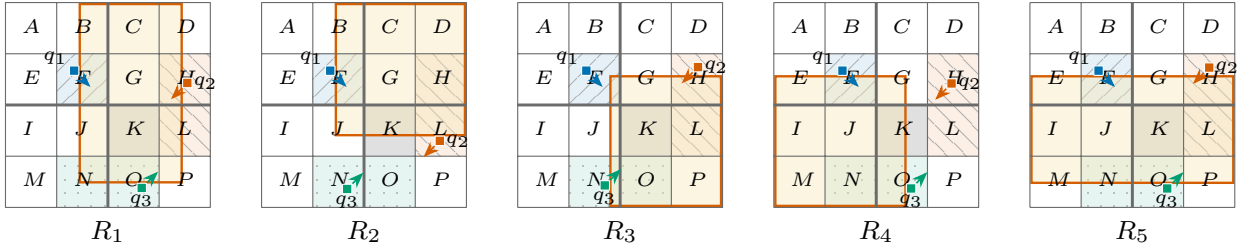

    Summing the masses of the cells that each rectangle can intersect gives
    \begin{align*}
        \mu\paren{R_1} &\leq B + C + D + F / 2 + G + X / 2 + J + Y / 2 + P, \\
        \mu\paren{R_2} &\leq B + C + D + F / 2 + G + H + J + L, \\
        \mu\paren{R_3} &\leq F / 2 + G + X / 2 + J + Y / 2 + P, \\
        \mu\paren{R_4} &\leq E + I + M + F / 2 + G + J + N + O, \\
        \mu\paren{R_5} &\leq E + I + M + F / 2 + G + X / 2 + J + Y / 2 + P.
    \end{align*}
    Moreover,
    \begin{equation*}
        X = 1 / 4 - D - P,
        \qquad
        Y = 1 / 4 - M - P,
        \qquad
        B + C + D = E + I + M = 1 / 4 - A.
    \end{equation*}
    Substituting these identities and using \eqref{eqn:corner-center} gives
    \begin{align*}
        \mu\paren{R_1}, \ \mu\paren{R_5}
        &\leq 1 / 2 - A + G + J + F / 2 - \frac{D + M}{2}, \\
        \mu\paren{R_2}
        &\leq 1 / 2 - A + G + J + F / 2 - D - P, \\
        \mu\paren{R_4}
        &\leq 1 / 2 - A + G + J + F / 2 - M - P, \\
        \mu\paren{R_3}
        &\leq 1 / 4 + F / 2 + G + J - \frac{D + M}{2}.
    \end{align*}
    The assumption of the lemma makes the first three displayed bounds at most $1 / 2$. Hence $R_1, R_2, R_4, R_5$ are safe.

    For $R_3$, since $A \leq 1 / 4$ and the assumption of the lemma holds,
    \begin{equation*}
        F + 2G + 2J < 2A \leq 1 / 2.
    \end{equation*}
    Thus, the last bound is also at most $1 / 2$.
    Consequently, every rectangle avoiding the three selected points has mass at most $1 / 2$.
\end{proof}

\medskip
\noindent
\textbf{The remaining cases.} We now assume
\begin{equation}
    \label{eqn:small-A}
    K = 0,
    \qquad F > 0,
    \qquad A \leq G + J + F / 2.
\end{equation}
The choice of points depends only on which of $G$ and $J$ are nonempty.

\begin{lemma}[Three nonempty center cells]
    \label{lem:three-center}
    If $G, J > 0$ and \eqref{eqn:small-A} holds, then there is a strong rectangular $1 / 2$-net of size at most three.
\end{lemma}

\begin{proof}
    Choose
    \begin{equation*}
        q_1 = p_F^{\mathrm{SE}},
        \qquad
        q_2 = \South\paren{G},
        \qquad
        q_3 = \East\paren{J},
    \end{equation*}
    where $q_1$ is balanced. Every rectangle avoiding the selected points, except one in the position shown in \cref{fig:three-center}, is safe because it either does not contain $p$ or is contained in the two central rows or the two central columns.

    In the remaining case, the rectangle is south of $q_2$ and east of $q_3$. It is contained in the region of mass $1 / 2 + A$ obtained by removing the first row and first column. It contains no point of $G$ or $J$, and it contains at most $F / 2$ mass from $F$. Hence, by \eqref{eqn:small-A},
    \begin{equation*}
        \mu\paren{R} \leq 1 / 2 + A - G - J - F / 2 \leq 1 / 2. \qedhere
    \end{equation*}
\end{proof}

\begin{figure}[H]
    \centering
    \input{figures/strong-net-three-center}
    \caption{The only remaining position of a rectangle avoiding the selected points when $F, G, J$ are nonempty and $K = 0$. Here $q_1$ is the balanced split point in $F$, $q_2 = \South\paren{G}$, and $q_3 = \East\paren{J}$.}
    \label{fig:three-center}
\end{figure}
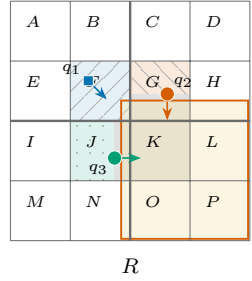

\begin{lemma}[Two adjacent nonempty center cells]
    \label{lem:adjacent-center}
    If $G > 0$, $J = 0$, and \eqref{eqn:small-A} holds, then there is a strong rectangular $1 / 2$-net of size at most three.
\end{lemma}

\begin{proof}
    Suppose that
    \begin{equation*}
        Y \coloneqq N \cup O
    \end{equation*}
    is nonempty. Without loss of generality, we may assume that $\North\paren{Y} \in N$. \Cref{lem:large-A} still covers the case $A > G + F / 2$, so assume
    \begin{equation*}
        A \leq G + F / 2.
    \end{equation*}
    The useful identities are
    \begin{equation*}
        A + C + D = F + N,
        \qquad
        A + B + D = G + O,
        \qquad
        A + E + M = L.
    \end{equation*}

    Set
    \begin{equation*}
        \alpha \coloneqq \max\set{A - G, 0},
        \qquad
        \beta \coloneqq F - \alpha = \min\set{F, F + G - A}.
    \end{equation*}
    Then $\alpha \leq F / 2 \leq \beta$, and
    \begin{equation}
        \label{eqn:adjacent-alpha-beta}
        \alpha + G \geq A,
        \qquad
        \beta + G \geq A,
        \qquad
        \beta \geq C - N.
    \end{equation}
    The last inequality follows from $C - N = F - A - D$. If $A \leq G$, then $\beta = F$. If $A > G$, then $\beta - \paren{C - N} = G + D$.

    Choose
    \begin{equation*}
        q_1 = p_F^{\mathrm{SE}}\paren{\alpha, \beta},
        \qquad
        q_2 = \South\paren{G},
        \qquad
        q_3 = \North\paren{Y}.
    \end{equation*}
    Any rectangle containing $p$ and avoiding the three selected points has one of seven possible positions. In three of them, it is contained in the middle two columns. The other four are shown in \cref{fig:adjacent-center}. Their masses satisfy
    \begin{align}
        \mu\paren{R_1} &\leq 1 / 2 + A - \alpha - G,
        &&\text{($q_1$ east, $q_2$ south),}\notag\\
        \mu\paren{R_2} &\leq 1 / 2 + A - \beta - G,
        &&\text{($q_1$ south, $q_2$ south, $q_3$ east),}\notag\\
        \mu\paren{R_3} &\leq 1 / 2 + D - \beta - N - O,
        &&\text{($q_1$ south, $q_2$ west, $q_3$ north),} \label{eqn:adj-R3}\\
        \mu\paren{R_4} &\leq 3 / 4 - \beta - G - N - O,
        &&\text{($q_1$ south, $q_2$ south, $q_3$ north).} \label{eqn:adj-R4}
    \end{align}
    The first two bounds are at most $1 / 2$ by \eqref{eqn:adjacent-alpha-beta}. For $R_3$,
    \begin{equation*}
        N + O - D = 1 / 4 + A - F - G \geq 0,
    \end{equation*}
    because $F + G \leq 1 / 4$, and hence the bound for $R_3$ in \eqref{eqn:adj-R3} is at most $1 / 2$. For $R_4$, the third inequality in \eqref{eqn:adjacent-alpha-beta} and the third-column equation give
    \begin{equation*}
        \beta + G + N + O
        \geq (C - N) + G + N + O
        = C + G + O
        = 1 / 4.
    \end{equation*}
    Thus, the bound for $R_4$ in \eqref{eqn:adj-R4} is also at most $1 / 2$.

    \begin{figure}[H]
        \centering
        \resizebox{\linewidth}{!}{\input{figures/strong-net-adjacent-center}}
        \caption{The four remaining rectangle positions in the adjacent-center case. Here $q_1$ is the split point in $F$, $q_2 = \South\paren{G}$, and $q_3 = \North\paren{Y}$. The points are deliberately offset so that no two selected points share an $x$- or $y$-coordinate.}
        \label{fig:adjacent-center}
    \end{figure}
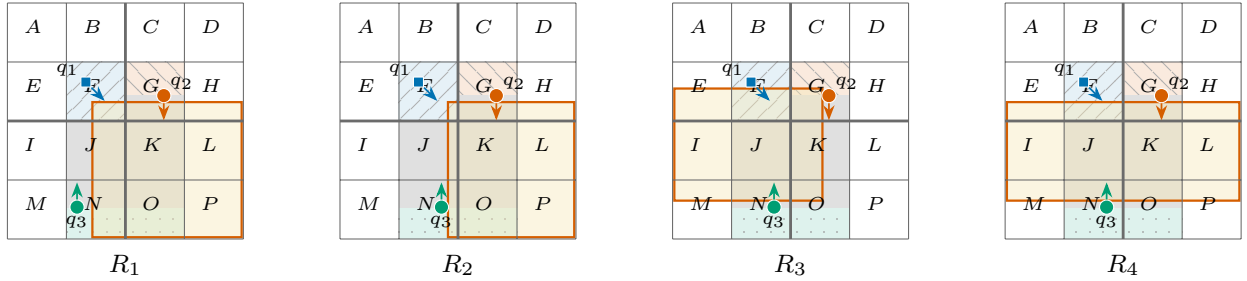

    It remains to consider $Y = 0$. The second and third columns and the first two rows give
    \begin{equation*}
        B + F = C + G = 1 / 4,
        \qquad
        A + D = F + G - 1 / 4,
        \qquad
        E + H = 1 / 4 - F - G.
    \end{equation*}
    Both expressions on the right are nonnegative, so
    \begin{equation*}
        F + G = 1 / 4,
        \qquad A = D = E = H = 0.
    \end{equation*}
    All remaining mass in the southern half is contained in $Z \coloneqq I \cup L \cup M \cup P$, where $Z = 1 / 2$.
    Choose
    \begin{equation*}
        q_1 = p_F^{\mathrm{SE}},
        \qquad
        q_2 = p_G^{\mathrm{SW}},
    \end{equation*}
    both balanced. Since $\varsigma_1^{\calR} = 3 / 4$ \cite[Thm.~1]{AshokAG14}, choose $q_3 \in Z$ such that every rectangle avoiding $q_3$ contains at most $3Z / 4 = 3 / 8$ mass from $Z$. Any rectangle avoiding $q_1, q_2, q_3$ that does not contain $p$ is safe. If it is east of $q_1$ and west of $q_2$, it lies in the middle two columns and is also safe. In every remaining case, it lies outside the first row, contains at most $\paren{F + G} / 2 = 1 / 8$ mass from $F \cup G$, and contains at most $3 / 8$ mass from $Z$. Thus, its total mass is at most $1 / 2$.
\end{proof}

\medskip
\noindent
\textbf{One nonempty center cell.} Assume from now on that
\begin{equation}
    \label{eqn:one-center-assumptions}
    G = J = K = 0,
    \qquad F > 0,
    \qquad A \leq F / 2.
\end{equation}
The identities reduce to
\begin{align}
    &A + D + M + P = F, \label{eqn:one-corners}\\
    &H + L = 1 / 4 - D - P = 1 / 4 - F + A + M, \label{eqn:X-mass}\\
    &N + O = 1 / 4 - M - P = 1 / 4 - F + A + D. \label{eqn:Y-mass}
\end{align}
Let
\begin{equation*}
    \tau \coloneqq \max\set{A, 1 / 4 - F}.
\end{equation*}
We have the following lemma.
\begin{lemma}
    \label{lem:se-anchor}
    Let $q_3 \in V$ lie southeast of $p$. Suppose there are sets $Z_E, Z_S \subseteq V$, each disjoint from the first row, the first column, and $F$, such that
    \begin{equation*}
        Z_E \subseteq \mathsf E\paren{q_3},
        \qquad
        Z_S \subseteq \mathsf S\paren{q_3},
        \qquad
        Z_E, Z_S \geq \tau
    \end{equation*}
    where we identify a set and its total mass. Then $\set{\East\paren{F}, \South\paren{F}, q_3}$ is a strong rectangular $1 / 2$-net.
\end{lemma}

\begin{proof}
    Set $q_1 = \East\paren{F}$ and $q_2 = \South\paren{F}$, possibly with $q_1 = q_2$. By the choice of extreme points, $F \subseteq \mathsf W\paren{q_1}$ and $F \subseteq \mathsf N\paren{q_2}$ (see \cref{fig:se-anchor}).

    Let $R$ be a rectangle avoiding $q_1, q_2, q_3$. If $R$ does not contain $p$, then $R$ is safe. We may therefore assume that $p \in R$. Since $q_1, q_2$ lie northwest of $p$, the rectangle $R$ is east or south of each of them. Since $q_3$ lies southeast of $p$, the rectangle $R$ is west or north of $q_3$. In the former case $R \cap Z_E = \varnothing$, and in the latter $R \cap Z_S = \varnothing$. Thus, in either case, there is a set $Z \in \set{Z_E, Z_S}$ such that
    \begin{equation*}
        R \cap Z = \varnothing, \qquad Z \geq \tau.
    \end{equation*}

    Suppose first that $R$ is either east of $q_1$ or south of $q_2$. In the first case, $R$ is outside the first column and misses all of $F$. In the second case, $R$ is outside the first row and again misses all of $F$. In either case, $R$ lies outside a row or column of mass $1 / 4$ and also misses the disjoint sets $F$ and $Z$. Thus,
    \begin{equation*}
        \mu\paren{R}
        \leq 3 / 4 - F - Z
        \leq 3 / 4 - F - \tau
        \leq 3 / 4 - F - \paren*{1 / 4 - F}
        = 1 / 2.
    \end{equation*}

    It remains to consider the case in which $R$ is neither east of $q_1$ nor south of $q_2$. Since $R$ avoids both points, it must be south of $q_1$ and east of $q_2$. It is therefore outside both the first row and the first column, whose complement has mass $1 / 2 + A$. Since $R$ misses $Z$,
    \begin{equation*}
        \mu\paren{R}
        \leq 1 / 2 + A - Z
        \leq 1 / 2 + A - \tau
        \leq 1 / 2.
    \end{equation*}
    Thus, every rectangle avoiding $q_1, q_2, q_3$ is safe.
\end{proof}

\begin{figure}[H]
    \centering
    \input{figures/strong-net-se-anchor}
    \caption{Illustration of \cref{lem:se-anchor}. In panel~(a), the gray portions of $F$ contain no point of $F$ by the extremality of $q_1$ and $q_2$. Panel~(b) is schematic because the sets $Z_E$ and $Z_S$ need not be rectangular. Their placement indicates only that $Z_E \subseteq \mathsf E\paren{q_3}$ and $Z_S \subseteq \mathsf S\paren{q_3}$. The proof also uses that they are disjoint from the first row, the first column, and $F$, and that each has mass at least $\tau$.}
    \label{fig:se-anchor}
\end{figure}
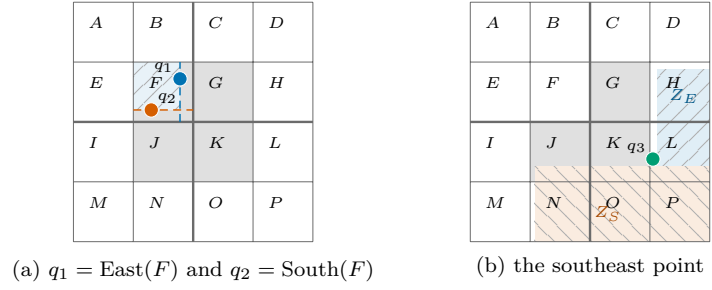

\begin{lemma}[One nonempty center cell]
    \label{lem:one-center}
    If \eqref{eqn:one-center-assumptions} holds, then there is a strong rectangular $1 / 2$-net of size at most three.
\end{lemma}

\begin{proof}
    Write
    \begin{equation*}
        X \coloneqq H \cup L,
        \qquad
        Y \coloneqq N \cup O.
    \end{equation*}

    We first consider the cases in which $X = 0$ or $Y = 0$. Suppose that $X = 0$. The fourth-column equation gives $D + P = 1 / 4$, and \eqref{eqn:one-corners} then gives $F = 1 / 4 + A + M$. Since $F \leq 1 / 4$, it follows that
    \begin{equation*}
        F = 1 / 4,
        \qquad
        A = M = 0.
    \end{equation*}
    The second-column equation now gives $N = 0$. Hence either $O > 0$, or $O = 0$ and the fourth-row equation gives $P = 1 / 4$. In either case, there is a point $q_3$ southeast of $p$. Here $\tau = 0$, so \cref{lem:se-anchor} applies with $Z_E = Z_S = \varnothing$. The case $Y = 0$ is symmetric.

    Assume now that $X, Y > 0$. We first consider the case $\West\paren{X} \in L$. Take
    \begin{equation*}
        q_3 = \West\paren{X},
        \qquad
        Z_E = X,
        \qquad
        Z_S = Y \cup P.
    \end{equation*}
    Then $q_3$ lies southeast of $p$, all of $X$ is east of $q_3$, and all of $Y \cup P$ is south of $q_3$. Moreover,
    \begin{equation*}
        X = \paren*{1 / 4 - F} + A + M \geq \tau
    \end{equation*}
    and, using \eqref{eqn:one-corners},
    \begin{equation*}
        Y + P = 1 / 4 - M = \paren*{1 / 4 - F} + A + D + P \geq \tau.
    \end{equation*}
    Thus \cref{lem:se-anchor} applies (see \cref{fig:one-center-easy} for an illustration).

    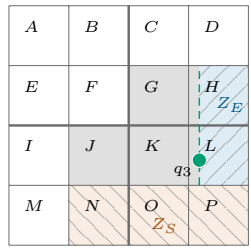
\begin{figure}[H]
        \centering
        \input{figures/strong-net-one-center-easy}
        \caption{The first case. If $\West\paren{X} \in L$, then $q_3 = \West\paren{X}$ has all of $X$ to its east and all of $Y \cup P$ to its south. The other case is obtained by reflection.}
        \label{fig:one-center-easy}
    \end{figure}

    Suppose next that $\North\paren{Y} \in O$. The symmetric choice is
    \begin{equation*}
        q_3 = \North\paren{Y},
        \qquad
        Z_S = Y,
        \qquad
        Z_E = X \cup P.
    \end{equation*}
    Indeed,
    \begin{equation*}
        Y = \paren*{1 / 4 - F} + A + D \geq \tau,
        \qquad
        X + P = 1 / 4 - D = \paren*{1 / 4 - F} + A + M + P \geq \tau,
    \end{equation*}
    so \cref{lem:se-anchor} again applies.

    We may therefore assume that $\West\paren{X} \in H$ and $\North\paren{Y} \in N$. Let
    \begin{equation*}
        \rho \coloneqq P \cap \mathsf E\paren{\West\paren{X}},
        \qquad
        \eta \coloneqq P \cap \mathsf S\paren{\North\paren{Y}}.
    \end{equation*}
    A rectangle west of $\West\paren{X}$ misses both $X$ and $\rho$, while a rectangle north of $\North\paren{Y}$ misses both $Y$ and $\eta$. First suppose that
    \begin{equation}
        \label{eqn:rich-tails}
        \mu\paren{\rho} + \mu\paren{\eta} \geq P - A.
    \end{equation}
    Set
    \begin{equation*}
        \alpha \coloneqq \max\set{A, D + P - \mu\paren{\rho}},
        \qquad
        \beta \coloneqq \max\set{A, M + P - \mu\paren{\eta}}.
    \end{equation*}
    Then $\alpha + \beta \leq F$. Indeed, according to which term attains each maximum, the difference $F - \paren{\alpha + \beta}$ is one of
    \begin{equation*}
        F - 2A,
        \qquad M + \mu\paren{\rho},
        \qquad D + \mu\paren{\eta},
        \qquad A - P + \mu\paren{\rho} + \mu\paren{\eta}.
    \end{equation*}
    All four quantities are nonnegative by \eqref{eqn:one-center-assumptions} and \eqref{eqn:rich-tails}. Choose
    \begin{equation*}
        q_1 = p_F^{\mathrm{SE}}\paren{\alpha, \beta},
        \qquad
        q_2 = \West\paren{X},
        \qquad
        q_3 = \North\paren{Y}.
    \end{equation*}

    Let $R$ avoid $q_1, q_2, q_3$. If $p \notin R$, then $R$ is safe, so assume that $p \in R$. If $R$ is east of $q_1$, then it is either west or south of $q_2$. These are the rectangles $R_1$ and $R_2$ in \cref{fig:one-center-rich}. If $R$ is south of $q_1$, then it is either north or east of $q_3$. These are the rectangles $R_3$ and $R_4$. Using \eqref{eqn:X-mass} and \eqref{eqn:Y-mass}, we obtain
    \begin{align*}
        \mu\paren{R_1}
        &\leq 3 / 4 - X - \mu\paren{\rho} - \alpha
        = 1 / 2 - \alpha + D + P - \mu\paren{\rho}, \\
        \mu\paren{R_2}
        &\leq 1 / 2 + A - \alpha, \\
        \mu\paren{R_3}
        &\leq 3 / 4 - Y - \mu\paren{\eta} - \beta
        = 1 / 2 - \beta + M + P - \mu\paren{\eta}, \\
        \mu\paren{R_4}
        &\leq 1 / 2 + A - \beta.
    \end{align*}
    The definitions of $\alpha$ and $\beta$ make all four right-hand sides at most $1 / 2$.

    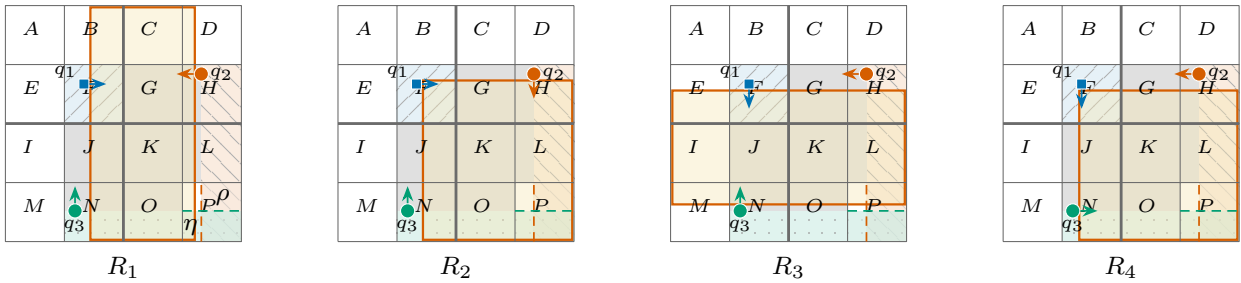
\begin{figure}[H]
        \centering
        \resizebox{\linewidth}{!}{\input{figures/strong-net-one-center-rich}}
        \caption{The four remaining rectangle positions in the one-center case. Here $q_1 = p_F^{\mathrm{SE}}\paren{\alpha, \beta}$, $q_2 = \West\paren{X}$, and $q_3 = \North\paren{Y}$. Within $P$, the diagonally hatched region is $\rho$, and the dotted region is $\eta$.}
        \label{fig:one-center-rich}
    \end{figure}

    Finally, suppose that
    \begin{equation*}
        \mu\paren{\rho} + \mu\paren{\eta} < P - A.
    \end{equation*}
    The definitions of $\rho$ and $\eta$ imply
    \begin{equation*}
        \mu\paren*{P \setminus \paren{\mathsf E\paren{\West\paren{X}} \cup \mathsf S\paren{\North\paren{Y}}}}
        \geq P - \mu\paren{\rho} - \mu\paren{\eta}
        > A.
    \end{equation*}
    This set is nonempty, so choose $q_3$ from it. Then $q_3$ lies southeast of $p$ and
    \begin{equation*}
        x\paren{q_3} < x\paren{\West\paren{X}}, \qquad y\paren{q_3} > y\paren{\North\paren{Y}}.
    \end{equation*}
    Thus $X \subseteq \mathsf E\paren{q_3}$ and $Y \subseteq \mathsf S\paren{q_3}$. Moreover, \eqref{eqn:X-mass} and \eqref{eqn:Y-mass} give
    \begin{equation*}
        X = \paren*{1 / 4 - F} + A + M \geq \tau,
        \qquad
        Y = \paren*{1 / 4 - F} + A + D \geq \tau.
    \end{equation*}
    \Cref{lem:se-anchor} therefore applies with $Z_E = X$ and $Z_S = Y$.
\end{proof}

\begin{figure}[H]
    \centering
    \input{figures/strong-net-one-center-poor}
    \caption{The last case of \cref{lem:one-center}. The highlighted part of $P$ lies west of $\West\paren{X}$ and north of $\North\paren{Y}$ and has mass greater than $A$. Choosing $q_3$ there places all of $X$ east of $q_3$ and all of $Y$ south of $q_3$.}
    \label{fig:one-center-poor}
\end{figure}
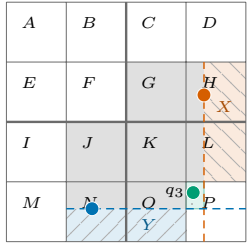

\medskip
\noindent
\begin{proof}[\textbf{Proof of \cref{thm:strong_R_3_nets}}]
    The preceding case analysis shows that every exact quartile grid has a strong rectangular $1 / 2$-net of size at most three.

    Let $V \subseteq \R^2$ be an arbitrary finite point set.
    Replace every point of $V$ by four sufficiently close distinct copies, obtaining a set $\bar{V}$.
    Choose the copies so that all their coordinates are distinct and every strict coordinate separation between the original points is preserved.
    Since $\card{\bar{V}} = 4\card V$, we can divide $\bar{V}$ into an exact quartile grid.
    The preceding case analysis then gives a set $\bar{Q} \subseteq \bar{V}$ of size at most three that intersects every rectangle containing more than half of the points of $\bar{V}$.

    Let $Q \subseteq V$ consist of the original points corresponding to the points of $\bar{Q}$.
    Suppose, for a contradiction, that an axis-parallel rectangle $R$ avoids $Q$ and contains the set $X \coloneqq R \cap V$ with $\card X > \card V / 2$.
    Let $R_X$ be the smallest closed axis-parallel rectangle containing $X$.
    Since $R_X \subseteq R$, every point of $Q$ lies outside $R_X$.
    The choice of the copies therefore ensures that the smallest axis-parallel rectangle containing all four copies of every point of $X$ avoids $\bar{Q}$.
    However, this rectangle contains
    \begin{equation*}
        4\card X > 2\card V = \frac{\card{\bar{V}}}{2}
    \end{equation*}
    points of $\bar{V}$, contradicting the choice of $\bar{Q}$.
    Therefore, $Q$ is a strong rectangular $1 / 2$-net for $V$ with $\card Q \leq 3$.
\end{proof}

\medskip
\noindent
\textbf{Lower bounds for strong rectangular nets.}
All three lower-bound instances are specified by a permutation $\pi$ describing the point set $\setst{\paren{i, \pi_i}}{i \in [n]}$.
For $k = 2, 3, 4$, we use the following permutations:
\begin{align*}
    &\pi^2 \coloneqq \paren{7, 8, 9, 12, 13, 14, 1, 2, 3, 11, 10, 4, 5, 6},\\
    &\pi^3 \coloneqq \paren{12, 13, 16, 9, 7, 19, 5, 21, 22, 23, 24, 1, 2, 25, 26, 3, 4, 20, 6, 18, 8, 17, 10, 11, 14, 15}, \text{ and }\\
    &\pi^4 \coloneqq \paren{9, 10, 11, 5, 4, 14, 15, 16, 1, 2, 3, 13, 12, 6, 7, 8}.
\end{align*}
We verify the lower bounds by exhaustive enumeration.
For $k = 2, 3, 4$, every $k$-point subset of the corresponding instance is avoided by an axis-parallel rectangle containing at least $8$, $12$, and $6$ points, respectively.
The complete verification is available in the accompanying code repository.\footnote{\url{https://github.com/gabrielmorete/Reproduction-Metric-Condorcet}}
Thus,
\begin{equation*}
    \varsigma_2^{\calR} \geq \frac{4}{7}, \qquad \varsigma_3^{\calR} \geq \frac{6}{13}, \qquad \varsigma_4^{\calR} \geq \frac{3}{8}.
\end{equation*}

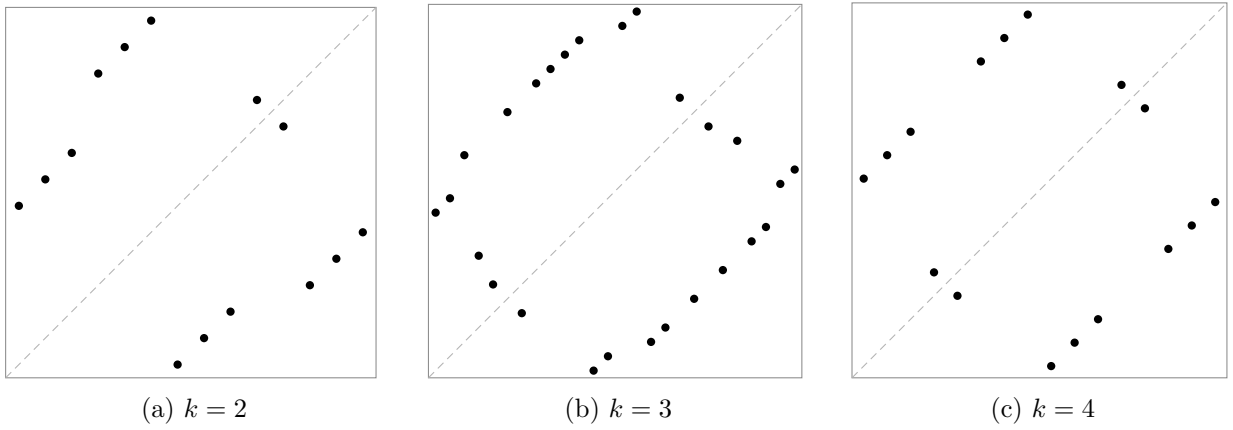
\begin{figure}[ht!]
    \centering
    \begin{subfigure}[t]{0.32\textwidth}
        \centering
        \input{figures/strong-net-lower-bound-k2}
        \caption{$k = 2$}
    \end{subfigure}
    \hfill
    \begin{subfigure}[t]{0.32\textwidth}
        \centering
        \input{figures/strong-net-lower-bound-k3}
        \caption{$k = 3$}
    \end{subfigure}
    \hfill
    \begin{subfigure}[t]{0.32\textwidth}
        \centering
        \input{figures/strong-net-lower-bound-k4}
        \caption{$k = 4$}
    \end{subfigure}
    \caption{Instances generating the lower bounds for strong rectangular nets.}
    \label{fig:strong-net-lower-bounds}
\end{figure}

\section{Deferred details for \nameref{sec:models}}
\label{apx:models}

\medskip
\noindent

\begin{proof}[\textbf{Proof of \cref{prp:voronoi-equivalence}}]
    First, by restricting the challenger and by considering $C = V$, we obtain $\pF\pF(k) \geq \pF\pR(k) \geq \pF\pP(k)$.
    It remains to prove $\pF\pF(k) \leq \pF\pP(k)$.
    The idea is to approximate a free-challenger instance by a sequence of instances in which the challenger is peer-restricted.
    We add increasingly many possible challenger locations as voters, while replacing each original voter by a much larger cluster of nearby voters.
    The added voter locations approximate the choices available to a free challenger, while the fraction of the electorate formed by these added voters tends to zero.
    Assume for contradiction that $\pF\pF(k) > \pF\pP(k)$.
    Then there exist a voter set $V = \set{v_1, \dots, v_n}$ and an integer $p \in \Z_{\geq 0}$ such that $\pF\pP(k) < p / n$ and, for every $S \subseteq \calX$ with $\card S \leq k$, there exists $c \in \calX$ satisfying
    \begin{equation*}
        \card{\setst{v \in V}{c \succ_v S}} \geq p.
    \end{equation*}

    Let $r \coloneqq 1 + \max_{j \in [n]} d\paren{v_1, v_j}$.
    Fix a countable dense set $\bar{\calC} \coloneqq \set{c_1, c_2, \dots} \subseteq \calX$, which exists because $\calX$ is finite-dimensional.
    For each $q \in \Z_{> 0}$ and each $i \in [n]$, take a set $U_{i, q}$ of $q^2$ points in $B\paren{v_i, 1 / q}$, with these sets pairwise disjoint.
    We take $V_q \coloneqq \set{c_1, \dots, c_q} \bigcup_{i = 1}^n U_{i, q}$ as the voter set of such an instance, which can be assumed to be in general position by introducing sufficiently small perturbations.

    For each $q$, let $\widetilde{S}_q$ be an optimal coalition chosen by the free defender for $V_q$.
    We obtain a coalition $S_q \subseteq B\paren{v_1, 2r}$ by replacing every point $x \in \widetilde{S}_q \setminus B\paren{v_1, 2r}$ with $v_1$.
    To justify this replacement, let $u \in U_{i, q}$ and $x \notin B\paren{v_1, 2r}$.
    Then,
    \begin{equation*}
        d\paren{u, v_1} \leq 1 / q + \max_{j \in [n]} d\paren{v_1, v_j} \leq r,
    \end{equation*}
    whereas the triangle inequality gives $d\paren{u, x} \geq d\paren{v_1, x} - d\paren{u, v_1} > r$.
    Thus, $v_1$ is closer than $x$ to every clustered voter.
    Replacing every such $x$ by $v_1$ therefore cannot increase the number of clustered voters captured by any challenger.
    The replacement can affect only the $q$ added voters from $\bar{\calC}$.
    Since $\widetilde{S}_q$ is optimal and $\card{V_q} = nq^2 + q$, it follows that
    \begin{equation*}
        \sup_{c \in V_q} \card{\setst{v \in V_q}{c \succ_v S_q}}
        \leq \pF\pP(k)\paren{nq^2 + q} + q.
    \end{equation*}
    We may assume without loss of generality that $\card{S_q} = k$, since adding points from $B\paren{v_1, 2r}$ to the defending coalition cannot increase the number of voters captured by any challenger.
    Write $S_q = \curly{s_1^q, \dots, s_k^q}$.
    By compactness of $B\paren{v_1, 2r}$, the Bolzano--Weierstrass theorem yields a subsequence along which $s_i^q \to s_i$ for every $i \in [k]$.
    Denote the resulting coalition by $S \coloneqq \curly{s_1, \dots, s_k}$.

    By the choice of $V$, there is a challenger $c \in \calX$ that captures at least $p$ voters against $S$.
    Choose $p$ of these voters.
    Since their preferences are strict and the distance function is continuous, the set of points in $\calX$ that all $p$ voters prefer to every member of $S$ is open and contains $c$.
    By density, this set contains some $\bar{c} \in \bar{\calC}$.
    For all sufficiently large $q$, the challenger may choose $\bar{c} \in V_q$.
    By the convergence of $S_q$ to $S$ and the definition of the clusters, continuity implies that $\bar{c}$ captures all $pq^2$ voters in the corresponding clusters.
    Applying the preceding bound to $\bar{c}$ gives
    \begin{equation*}
        pq^2 \leq \pF\pP(k)\paren{nq^2 + q} + q.
    \end{equation*}

    Dividing by $q^2$ and letting $q$ tend to infinity gives $p / n \leq \pF\pP(k)$, a contradiction.
\end{proof}

\section{Deferred details for \nameref{sec:applications}}
\label{apx:applications}

\medskip
\noindent
\textbf{Proof of \cref{lem:giveaway} and a counterexample to its higher-dimensional analogue.} We first prove the lemma.

\begin{proof}[Proof of \cref{lem:giveaway}]
    After a translation, we may assume that $a = 0$.
    Fix a distance $d$ induced by a norm, and let $B = B_d(0, 1)$ be its unit ball.
    For a vector $x \in \R^2$, define $\R x \coloneqq \setst{\alpha x}{\alpha \in \R}$.

    A line supports $B$ if it intersects $B$ and $B$ is entirely contained in one of the two closed half-planes determined by the line.
    Choose a supporting line of $B$ parallel to $b$, and let $u \in \partial B$ be a point where it meets $B$ (see \cref{fig:giveaway_proof}).
    Since $B$ is centrally symmetric, $-u + \R b$ also supports $B$.
    Moreover, $u$ and $b$ are linearly independent.
    Otherwise, $u + \R b$ would contain the origin, which lies in the interior of $B$, contradicting that the line supports $B$.

    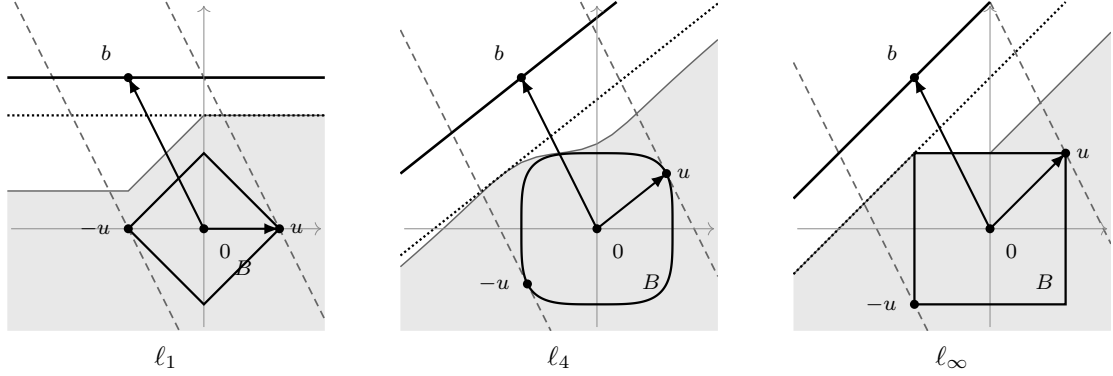
\begin{figure}[ht!]
        \centering
        \input{figures/giveaway}
        \caption{Illustration of the proof of \cref{lem:giveaway} for $b = \paren{-1, 2}$ under the $\ell_1$, $\ell_4$, and $\ell_\infty$ norms.
        The shaded region is $\Dom(0, b)$.
        The dotted line is the boundary of $\conv\paren{\Dom(0, b)}$.
        The dashed lines are the supporting lines $u + \R b$ and $-u + \R b$ of $B$.
        The solid line is $b + \R u$.}
        \label{fig:giveaway_proof}
    \end{figure}

    We show that
    \begin{equation}
        \label{eqn:giveaway_proof}
        \Dom(0, b) \subseteq \setst{\alpha u + \beta b \in \R^2}{\beta < 1, \alpha, \beta \in \R}.
    \end{equation}
    For each $\alpha \in \R$, the point $y \coloneqq \alpha u$ is a closest point to the origin on the line $y + \R b$.
    For $\alpha = 0$, the claim is trivial.
    If $\alpha > 0$, then $y + \R b$ supports the ball $\alpha B = B(0, \alpha)$, so every point on this line has norm at least $\alpha = d\paren{y, 0}$.
    If $\alpha < 0$, the same conclusion follows from the supporting line through $-u$.

    Consider the closed half-plane bounded by $b + \R u$ that does not contain the origin.
    Every point in this half-plane can be written as $x = y + \beta b$ for some $\beta \geq 1$ and $y \in \R u$.
    Since $d$ is induced by a norm, $d\paren{y + \gamma b, 0}$ is a convex function of $\gamma$.
    By the preceding paragraph, this function is minimized at $\gamma = 0$ and is therefore nondecreasing for $\gamma \geq 0$.
    Consequently,
    \begin{equation*}
        d\paren{x, b} = d\paren{y + \paren{\beta - 1}b, 0} \leq d\paren{y + \beta b, 0} = d\paren{x, 0}.
    \end{equation*}
    Thus, every point in this half-plane is at least as close to $b$ as it is to $0$ and therefore does not belong to $\Dom(0, b)$.
    This proves \eqref{eqn:giveaway_proof}.
    The half-plane on the right-hand side of \eqref{eqn:giveaway_proof} is convex, contains $\conv\paren{\Dom(0, b)}$, and does not contain $b$.
    Therefore, $b \notin \conv\paren{\Dom(0, b)}$.

    For the $\ell_1$ norm, we can make the choice of $u$ more specific.
    The vertices of $B$ are $\paren{1, 0}$, $\paren{0, 1}$, $\paren{-1, 0}$, and $\paren{0, -1}$, which lie on the coordinate axes (see \cref{fig:giveaway_proof} for an illustration).
    Every supporting line of $B$ contains a vertex, so we may take $u$ to be one of them.
    Hence, $\R u$ is a coordinate axis and $b + \R u$ is axis-parallel.
    By \eqref{eqn:giveaway_proof}, $\Dom_{\ell_1}(0, b)$ is contained in the open half-plane bounded by this line.
    This completes the proof.
\end{proof}

We next show that the analogue of \cref{lem:giveaway} in $\R^3$ fails for every $\ell_p$ norm with $p \neq 2$.
Let $a \coloneqq \paren{0, 0, 0}$ and $b \coloneqq \paren{1, 1, 1}$.

For $1 \leq p < 2$, consider the three cyclic permutations of
\begin{equation*}
    x_t \coloneqq \paren{1 + 2t, 1 - t, 1 - t}.
\end{equation*}
Their average is $b$.
Moreover,
\begin{align*}
    \lim_{t \to \infty}
    \frac{\norm{x_t - a}_p^p - \norm{x_t - b}_p^p}{t^{p - 1}}
    &= \lim_{t \to \infty}
    \frac{\paren{1 + 2t}^p + 2\paren{t - 1}^p - \paren{2^p + 2}t^p}{t^{p - 1}} \\
    &= \lim_{t \to \infty}
    t\paren*{\paren{2 + 1 / t}^p + 2\paren{1 - 1 / t}^p - \paren{2^p + 2}}, \\
    \intertext{taking $s = 1 / t$, we obtain}
    &= \lim_{s \to 0^+}
    \frac{\paren{2 + s}^p + 2\paren{1 - s}^p - \paren{2^p + 2}}{s} \\
    &= \lim_{s \to 0^+}
    \paren*{p\paren{2 + s}^{p - 1} - 2p\paren{1 - s}^{p - 1}} \\
    &= p\paren{2^{p - 1} - 2} < 0.
\end{align*}
Thus, for sufficiently large $t$, all three points belong to $\Dom(a, b)$ and $b \in \conv\paren{\Dom(a, b)}$.

For $2 < p < \infty$, the same argument applies to the cyclic permutations of $y_t \coloneqq \paren{1 - 2t, 1 + t, 1 + t}$.
When $p = \infty$, this construction works for $t > 1$, since $\norm{y_t - a}_\infty = \max\set{2t - 1, t + 1} < 2t = \norm{y_t - b}_\infty$.

\medskip
\noindent
\textbf{Nearest-candidate rounding does not extend to other $\bm{\ell_p}$ norms.}
The instance in \cref{fig:rounding_counterexample} shows that, already on the plane, the nearest-candidate rounding argument used in the proof of \cref{lem:euclidean_projection} fails for every $\ell_p$ norm with $p \neq 2$.
Since $q$ is the midpoint of $v_1$ and $v_2$, every convex set containing more than half of the voters contains $q$, and hence $q$ is a weak $1 / 2$-net for $V$.
For each of $q$, $v_1$, and $v_2$, the squared Euclidean distances to $a$ and $b$ are equal.
At $q$, the two squared coordinate differences to $a$ are more unequal than those to $b$, whereas at each voter the two squared coordinate differences to $b$ are more unequal than those to $a$.
Strict concavity of the function $x^{p / 2}$ therefore makes $q$ closer to $a$ and both voters closer to $b$ for $1 \leq p < 2$, while strict convexity reverses all comparisons for $2 < p < \infty$.
For $p = \infty$, the same follows by comparing the largest coordinate differences.
Thus, for every $p \neq 2$, rounding $q$ to its closest candidate allows the other candidate to capture both voters.

\begin{figure}[h!]
    \centering
    \input{figures/rounding-counterexample}
    \caption{Example showing that the rounding argument does not extend to other $\ell_p$ norms.
    The three panels use $\ell_{2 - \delta}$, $\ell_2$, and $\ell_{2 + \delta}$ with $\delta = 1 / 2$.
    The voters are $v_1 = \paren{2, 1}$ and $v_2 = \paren{6, 9}$, their midpoint is $q = \paren{4, 5}$, and the candidates are $a = \paren{9, 5}$ and $b = \paren{1, 9}$.
    The shaded region is $\Dom(a, b)$.}
    \label{fig:rounding_counterexample}
\end{figure}
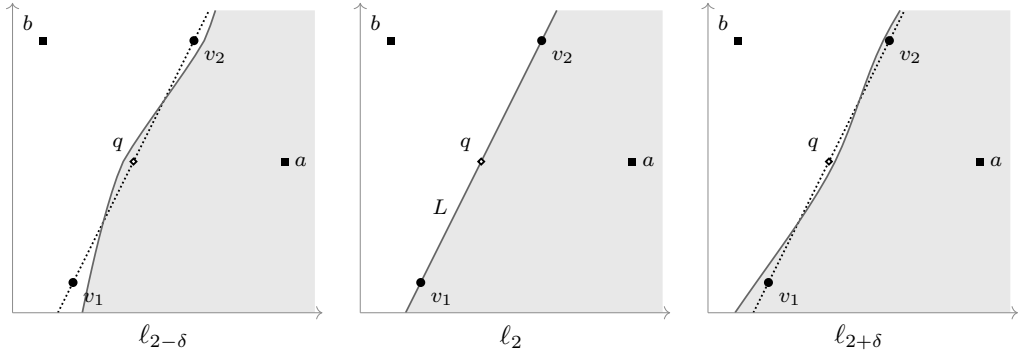

\end{document}

%% file: figures/strong-net-figure-styles.tex
\definecolor{gridblue}{HTML}{DCEAF7}
\definecolor{gridgold}{HTML}{F7E8B5}
\definecolor{pointblue}{HTML}{0072B2}
\definecolor{pointred}{HTML}{D55E00}
\definecolor{pointgreen}{HTML}{009E73}

\tikzset{
    bluepoint/.style = {circle, fill = pointblue, draw = white, line width = 0.35pt, inner sep = 2.15pt},
    redpoint/.style = {circle, fill = pointred, draw = white, line width = 0.35pt, inner sep = 2.15pt},
    greenpoint/.style = {circle, fill = pointgreen, draw = white, line width = 0.35pt, inner sep = 2.15pt},
    bluesplit/.style = {rectangle, rounded corners = 0.25pt, fill = pointblue, draw = white, line width = 0.35pt, inner sep = 2.15pt},
    redsplit/.style = {rectangle, rounded corners = 0.25pt, fill = pointred, draw = white, line width = 0.35pt, inner sep = 2.15pt},
    greensplit/.style = {rectangle, rounded corners = 0.25pt, fill = pointgreen, draw = white, line width = 0.35pt, inner sep = 2.15pt},
    bluehatch/.style = {postaction = {pattern = {Lines[angle = 45, distance = 5pt, line width = 0.18pt]}, pattern color = black!32}},
    redhatch/.style = {postaction = {pattern = {Lines[angle = -45, distance = 5pt, line width = 0.18pt]}, pattern color = black!32}},
    greenhatch/.style = {postaction = {pattern = {Dots[distance = 5pt, radius = 0.32pt]}, pattern color = black!38}},
    avoidarrow/.style = {-{Stealth[length = 1.7mm]}, line width = 0.72pt},
    dangerrect/.style = {draw = pointred, fill = gridgold, fill opacity = 0.40, draw opacity = 1, line width = 0.72pt},
    emptycell/.style = {fill = black!12, draw = none},
    focuszone/.style = {fill = pointgreen!12, draw = pointgreen!48, line width = 0.45pt, greenhatch},
    guide/.style = {densely dashed, line width = 0.55pt}
}

\newcommand{\drawquartilegrid}{%
    \draw[step = 1, draw = black!58, line width = 0.34pt] (0, 0) grid (4, 4);
    \draw[draw = black!58, line width = 0.9pt] (2, 0) -- (2, 4);
    \draw[draw = black!58, line width = 0.9pt] (0, 2) -- (4, 2);
    \foreach \name/\x/\y in {
        A/0.08/3.92, B/1.08/3.92, C/2.08/3.92, D/3.08/3.92,
        E/0.08/2.92, F/1.08/2.92, G/2.08/2.92, H/3.08/2.92,
        I/0.08/1.92, J/1.08/1.92, K/2.08/1.92, L/3.08/1.92,
        M/0.08/0.92, N/1.08/0.92, O/2.08/0.92, P/3.08/0.92}
        \node[text = black, font = \tiny, anchor = north west] at (\x, \y) {$\name$};
}

%% file: figures/quartile-grid.tex
\begin{tikzpicture}[
    x = 0.83cm,
    y = 0.83cm,
    every node/.style = {font = \small}
]
    \begin{scope}
        \fill[black!9] (1, 1) rectangle (3, 3);
        \draw[step = 1, draw = black!58, line width = 0.34pt] (0, 0) grid (4, 4);
        \draw[draw = black!58, line width = 0.9pt] (2, 0) -- (2, 4);
        \draw[draw = black!58, line width = 0.9pt] (0, 2) -- (4, 2);
        \foreach \name/\x/\y in {
            A/0.08/3.92, B/1.08/3.92, C/2.08/3.92, D/3.08/3.92,
            E/0.08/2.92, F/1.08/2.92, G/2.08/2.92, H/3.08/2.92,
            I/0.08/1.92, J/1.08/1.92, K/2.08/1.92, L/3.08/1.92,
            M/0.08/0.92, N/1.08/0.92, O/2.08/0.92, P/3.08/0.92
        } {
            \node[text = black, font = \tiny, anchor = north west] at (\x, \y) {$\name$};
        }
        \node[circle, fill = black, inner sep = 1.25pt, label = {[label distance = -1.5pt, font = \scriptsize]below left:$p$}] at (2, 2) {};
    \end{scope}

    \begin{scope}[xshift = 5.5cm]
        \fill[black!3] (1.65, 0) rectangle (4, 4);
        \fill[black!3] (0, 0) rectangle (4, 2.35);
        \begin{scope}
            \clip (1.65, 0) rectangle (4, 4);
            \foreach \offset in {-4, -3.65, ..., 2.65} {
                \draw[draw = black!28, line width = 0.25pt] (0, \offset) -- (4, {4 + \offset});
            }
        \end{scope}
        \begin{scope}
            \clip (0, 0) rectangle (4, 2.35);
            \foreach \offset in {0, 0.35, ..., 6.65} {
                \draw[draw = black!28, line width = 0.25pt] (0, \offset) -- (4, {\offset - 4});
            }
        \end{scope}
        \draw[draw = black!38, line width = 0.3pt] (1.65, 0) rectangle (4, 4);
        \draw[draw = black!38, line width = 0.3pt] (0, 0) rectangle (4, 2.35);
        \draw[step = 1, draw = black!58, line width = 0.34pt] (0, 0) grid (4, 4);
        \draw[draw = black!58, line width = 0.9pt] (2, 0) -- (2, 4);
        \draw[draw = black!58, line width = 0.9pt] (0, 2) -- (4, 2);
        \foreach \name/\x/\y in {
            A/0.08/3.92, B/1.08/3.92, C/2.08/3.92, D/3.08/3.92,
            E/0.08/2.92, F/1.08/2.92, G/2.08/2.92, H/3.08/2.92,
            I/0.08/1.92, J/1.08/1.92, K/2.08/1.92, L/3.08/1.92,
            M/0.08/0.92, N/1.08/0.92, O/2.08/0.92, P/3.08/0.92
        } {
            \node[text = black, font = \tiny, anchor = north west] at (\x, \y) {$\name$};
        }
        \node[circle, fill = black, inner sep = 1.25pt, label = {[label distance = -1.5pt, font = \scriptsize]below left:$p$}] at (2, 2) {};
        \node[circle, fill = black, inner sep = 1.25pt] (q) at (1.65, 2.35) {};
        \node[anchor = west, fill = white, inner sep = 0.6pt, font = \scriptsize] at (1.73, 2.35) {$q$};

        \draw[draw = black!65, line width = 0.45pt] (4.08, 2.35) -- (4.20, 2.35) -- (4.20, 4) -- (4.08, 4);
        \node[right, font = \scriptsize] at (4.20, 3.175) {$\leq 1 / 2$};
        \draw[draw = black!65, line width = 0.45pt] (0, -0.08) -- (0, -0.20) -- (1.65, -0.20) -- (1.65, -0.08);
        \node[below, font = \scriptsize] at (0.825, -0.20) {$\leq 1 / 2$};
    \end{scope}
\end{tikzpicture}

%% file: figures/hierarchy.tex
\begin{tikzpicture}[
    >=Latex,
    single/.style={align=center, inner sep=3pt, font=\small},
    order/.style={->, line width=0.8pt, shorten <=3pt, shorten >=3pt},
    geometric/.style={->, dashed, line width=0.8pt, shorten <=3pt, shorten >=3pt},
    note/.style={align=left, inner sep=0pt, font=\scriptsize}
]
    \coordinate (wn-old) at (-5.00, 6.00);
    \node[single] (wn) at (-2.50,6.00) {\emph{Weak $\varepsilon$-nets}};
    \node[single] (srn) at (2.50,6.00) {\emph{Strong rectangular} \emph{$\varepsilon$-nets}};
    \node[single] (rr) at (-2.50,4.00) {\emph{General} \emph{metric} \emph{elections}\\$\pR\pR$};
    \node[single] (sp) at (2.50,4.00) {\emph{Strong peer selection}\\$\pP\pF, \pP\pR$};
    \node[single] (pp) at (0,2.00) {\emph{Peer selection}\\$\pP\pP$};
    \node[single] (vg) at (0,0) {\emph{Voronoi games}\\$\pF\pF, \pF\pR, \pF\pP$};

    \draw[geometric] (wn.south) -- (rr.north);
    \node[note, anchor=west] at (-2.00,5.25) {$\paren{\R^r, \ell_2}$};
    \node[note, anchor=west] at (-2.00,4.92) {\cref{lem:euclidean_projection}};

    \draw[geometric] (srn.south) -- (sp.north);
    \node[note, anchor=west] at (3.00,5.28) {$\paren{\R^2, \ell_1}$\\$\paren{\R^2, \ell_\infty}$};
    \node[note, anchor=west] at (3.00,4.72) {\cref{lem:strong_nets_to_peer}};

    \coordinate (wn-vg-turn) at (wn-old |- wn.west);
    \coordinate (wn-vg-corner) at (wn-old |- vg.west);
    \draw[geometric] (wn.west) -- (wn-vg-turn) -- (wn-vg-corner) -- (vg.west);
    \node[note, anchor=east, xshift=-4pt] at ($(wn-vg-turn)!0.5!(wn-vg-corner)+(0,0.22)$) {$\paren{\R^2, d}$};
    \node[note, anchor=east, xshift=-4pt] at ($(wn-vg-turn)!0.5!(wn-vg-corner)+(0,-0.22)$) {\cref{prop:weak_nets_to_voronoi}};

    \draw[order] (rr.south) -- (pp.north west);
    \node[note, anchor=south east] at (0,2.95) {Equation~\eqref{eqn:hierarchy_condorcet}};

    \draw[order] (sp.south) -- (pp.north east);
    \node[note, anchor=south west] at (2.25,2.95) {Equation~\eqref{eqn:hierarchy_strong}};

    \draw[order] (pp) -- (vg);
    \node[note, anchor=west] at (0.20,1.00) {Equation~\eqref{eqn:hierarchy_strong}};
\end{tikzpicture}

%% file: figures/manhattan-one-candidate.tex
\begin{tikzpicture}[
    x = 0.72cm,
    y = 0.72cm,
    axis/.style = {->, draw = black!45, line width = 0.35pt},
    candidate orbit/.style = {draw = black!48, line width = 0.55pt},
    voter orbit/.style = {densely dotted, draw = black!65, line width = 0.7pt},
    voter/.style = {circle, fill = black, inner sep = 1.2pt},
    candidate/.style = {rectangle, fill = black, inner sep = 1.5pt},
    every node/.style = {font = \scriptsize}
]
    \draw[axis] (-5, 0) -- (5.15, 0);
    \draw[axis] (0, -5) -- (0, 5.15);

    \draw[candidate orbit] (2, 1) -- (-1, 2) -- (-2, -1) -- (1, -2) -- cycle;
    \draw[voter orbit] (3, 4) -- (-4, 3) -- (-3, -4) -- (4, -3) -- cycle;
    \node[candidate] at (2, 1) {};
    \node[below right] at (2, 1) {$c_1$};
    \node[candidate] at (-1, 2) {};
    \node[above left] at (-1, 2) {$c_2$};
    \node[candidate] at (-2, -1) {};
    \node[above left] at (-2, -1) {$c_3$};
    \node[candidate] at (1, -2) {};
    \node[below right] at (1, -2) {$c_4$};

    \node[voter] at (3, 4) {};
    \node[above right] at (3, 4) {$v_1$};
    \node[voter] at (-4, 3) {};
    \node[above left] at (-4, 3) {$v_2$};
    \node[voter] at (-3, -4) {};
    \node[below left] at (-3, -4) {$v_3$};
    \node[voter] at (4, -3) {};
    \node[below right] at (4, -3) {$v_4$};
\end{tikzpicture}

%% file: figures/undecided-voter-example.tex
\begin{tikzpicture}[
    x = 1.5cm,
    y = 1.5cm,
    axis/.style = {->, draw = black!40, line width = 0.35pt},
    ballone/.style = {densely dotted, draw = black, line width = 0.8pt},
    balltwo/.style = {densely dashed, draw = black!70, line width = 0.75pt},
    ballinf/.style = {draw = black, line width = 0.75pt},
    trajectory/.style = {-{Latex[length = 1.6mm]}, draw = black!55, line width = 0.55pt},
    voter/.style = {circle, fill = black, inner sep = 1.2pt},
    candidate/.style = {rectangle, fill = black, inner sep = 1.5pt},
    every node/.style = {font = \scriptsize}
]
    \draw[axis] (-2.25, 0) -- (2.3, 0);
    \draw[axis] (0, -2.25) -- (0, 2.1);

    \draw[ballinf] (-1, -1) rectangle (1, 1);
    \draw[balltwo] (0, 0) circle[radius = 1];
    \draw[ballone] (1, 0) -- (0, 1) -- (-1, 0) -- (0, -1) -- cycle;

    \draw[trajectory] (-0.5, -0.5) -- (-1, -1);

    \node[candidate, label = right:$a$] at (1, 0) {};
    \node[candidate, label = above:$b$] at (0, 1) {};
    \node[candidate] at (-1, -1) {};
    \node[below left] at (-1, -1) {$c$};

    \node[voter] at (2, 0.25) {};
    \node[right] at (2, 0.25) {$v_1$};
    \node[voter] at (-2, 1.75) {};
    \node[above right] at (-2, 1.75) {$v_2$};
    \node[voter] at (0.25, -2) {};
    \node[right] at (0.25, -2) {$v_3$};
\end{tikzpicture}

%% file: figures/strong-net-full-center.tex
\begin{tikzpicture}[
  x=.66cm,y=.66cm,font=\scriptsize,
  bluepoint/.append style={inner sep=1.75pt},
  redpoint/.append style={inner sep=1.75pt},
  greensplit/.append style={inner sep=1.75pt},
  avoidarrow/.append style={line width=.58pt},
  every label/.append style={font=\tiny,inner sep=.2pt}]
  \begin{scope}
    \fill[pointblue!13,bluehatch] (2.34,2) rectangle (3,3);
    \fill[emptycell] (2,2) rectangle (2.34,3);
    \fill[pointred!13,redhatch] (1,1) rectangle (2,1.47);
    \fill[emptycell] (1,1.47) rectangle (2,2);
    \fill[pointgreen!10,greenhatch] (2,1) rectangle (3,2);
    \path[dangerrect] (.03,1.58) rectangle (2.23,3.97);
    \drawquartilegrid
    \node[bluepoint,label=above right:$q_1$] (g1) at (2.34,2.63) {};
    \node[redpoint,label=below left:$q_2$] (j1) at (1.63,1.47) {};
    \node[greensplit,label=below right:$q_3$] (k1) at (2.58,1.42) {};
    \draw[avoidarrow,pointblue] (g1)--+(-.43,0);
    \draw[avoidarrow,pointred] (j1)--+(0,.43);
    \draw[avoidarrow,pointgreen] (k1)--+(-.32,.32);
    \node at (2,-.43) {(a) corner $P$};
  \end{scope}
  \begin{scope}[xshift=3.72cm]
    \fill[pointblue!13,bluehatch] (1,2) rectangle (1.66,3);
    \fill[emptycell] (1.66,2) rectangle (2,3);
    \fill[pointred!13,redhatch] (2,1) rectangle (3,1.43);
    \fill[emptycell] (2,1.43) rectangle (3,2);
    \fill[pointgreen!10,greenhatch] (1,1) rectangle (2,2);
    \path[dangerrect] (1.77,1.55) rectangle (3.97,3.97);
    \drawquartilegrid
    \node[bluepoint,label=above left:$q_1$] (f2) at (1.66,2.70) {};
    \node[redpoint,label=below right:$q_2$] (k2) at (2.58,1.43) {};
    \node[greensplit,label=below left:$q_3$] (j2) at (1.42,1.42) {};
    \draw[avoidarrow,pointblue] (f2)--+(.43,0);
    \draw[avoidarrow,pointred] (k2)--+(0,.43);
    \draw[avoidarrow,pointgreen] (j2)--+(.32,.32);
    \node at (2,-.43) {(b) corner $M$};
  \end{scope}
  \begin{scope}[xshift=7.44cm]
    \fill[pointblue!13,bluehatch] (1,2.51) rectangle (2,3);
    \fill[emptycell] (1,2) rectangle (2,2.51);
    \fill[pointred!13,redhatch] (2.34,1) rectangle (3,2);
    \fill[emptycell] (2,1) rectangle (2.34,2);
    \fill[pointgreen!10,greenhatch] (2,2) rectangle (3,3);
    \path[dangerrect] (.03,.03) rectangle (2.23,2.40);
    \drawquartilegrid
    \node[bluepoint,label=above left:$q_1$] (f3) at (1.48,2.51) {};
    \node[redpoint,label=below right:$q_2$] (k3) at (2.34,1.36) {};
    \node[greensplit,label=above right:$q_3$] (g3) at (2.58,2.58) {};
    \draw[avoidarrow,pointblue] (f3)--+(0,-.43);
    \draw[avoidarrow,pointred] (k3)--+(-.43,0);
    \draw[avoidarrow,pointgreen] (g3)--+(-.32,-.32);
    \node at (2,-.43) {(c) corner $D$};
  \end{scope}
  \begin{scope}[xshift=11.16cm]
    \fill[pointblue!13,bluehatch] (2,2.51) rectangle (3,3);
    \fill[emptycell] (2,2) rectangle (3,2.51);
    \fill[pointred!13,redhatch] (1,1) rectangle (1.66,2);
    \fill[emptycell] (1.66,1) rectangle (2,2);
    \fill[pointgreen!10,greenhatch] (1,2) rectangle (2,3);
    \path[dangerrect] (1.77,.03) rectangle (3.97,2.40);
    \drawquartilegrid
    \node[bluepoint,label=above right:$q_1$] (g4) at (2.58,2.51) {};
    \node[redpoint,label=below left:$q_2$] (j4) at (1.66,1.36) {};
    \node[greensplit,label=above left:$q_3$] (f4) at (1.42,2.58) {};
    \draw[avoidarrow,pointblue] (g4)--+(0,-.43);
    \draw[avoidarrow,pointred] (j4)--+(.43,0);
    \draw[avoidarrow,pointgreen] (f4)--+(.32,-.32);
    \node at (2,-.43) {(d) corner $A$};
  \end{scope}
\end{tikzpicture}

%% file: figures/strong-net-opposite-center.tex
\begin{tikzpicture}[
  x=.792cm,y=.792cm,font=\scriptsize,
  bluesplit/.append style={inner sep=1.9pt},
  redsplit/.append style={inner sep=1.9pt},
  greensplit/.append style={inner sep=1.9pt},
  avoidarrow/.append style={line width=.62pt},
  every label/.append style={font=\tiny,inner sep=.2pt}]
  \begin{scope}
    \path[emptycell] (2,2) rectangle (3,3);
    \path[emptycell] (1,1) rectangle (2,2);
    \fill[pointblue!10,bluehatch] (1,2) rectangle (2,3);
    \fill[pointred!10,redhatch] (2,1) rectangle (3,2);
    \path[focuszone] (2,3) rectangle (4,4);
    \path[focuszone] (3,2) rectangle (4,3);
    \path[dangerrect] (1.39,1.39) rectangle (3.19,3.97);
    \drawquartilegrid
    \node[bluesplit,label=above left:$q_1$] (fo1) at (1.28,2.66) {};
    \node[redsplit,label=below right:$q_2$] (ko1) at (2.72,1.28) {};
    \node[greensplit,label=above right:$q_3$] (qo1) at (3.30,3.02) {};
    \draw[avoidarrow,pointblue] (fo1)--+(.34,-.34);
    \draw[avoidarrow,pointred] (ko1)--+(-.34,.34);
    \draw[avoidarrow,pointgreen] (qo1)--+(-.34,-.34);
    \node at (2,-.43) {(a) west of $q_3$};
  \end{scope}
  \begin{scope}[xshift=5.0cm]
    \path[emptycell] (2,2) rectangle (3,3);
    \path[emptycell] (1,1) rectangle (2,2);
    \fill[pointblue!10,bluehatch] (1,2) rectangle (2,3);
    \fill[pointred!10,redhatch] (2,1) rectangle (3,2);
    \path[focuszone] (2,3) rectangle (4,4);
    \path[focuszone] (3,2) rectangle (4,3);
    \path[dangerrect] (1.39,1.39) rectangle (3.97,2.91);
    \drawquartilegrid
    \node[bluesplit,label=above left:$q_1$] (fo2) at (1.28,2.66) {};
    \node[redsplit,label=below right:$q_2$] (ko2) at (2.72,1.28) {};
    \node[greensplit,label=above right:$q_3$] (qo2) at (3.30,3.02) {};
    \draw[avoidarrow,pointblue] (fo2)--+(.34,-.34);
    \draw[avoidarrow,pointred] (ko2)--+(-.34,.34);
    \draw[avoidarrow,pointgreen] (qo2)--+(-.34,-.34);
    \node at (2,-.43) {(b) south of $q_3$};
  \end{scope}
\end{tikzpicture}

%% file: figures/strong-net-outward.tex
\begin{tikzpicture}[
  x=.55cm,y=.55cm,font=\scriptsize,
  bluesplit/.append style={inner sep=1.45pt},
  redsplit/.append style={inner sep=1.45pt},
  greensplit/.append style={inner sep=1.45pt},
  avoidarrow/.append style={line width=.53pt},
  every label/.append style={font=\tiny,inner sep=.15pt}]
  \begin{scope}
    \path[emptycell] (2,1) rectangle (3,2);
    \fill[pointblue!10,bluehatch] (1,2) rectangle (2,3);
    \fill[pointred!10,redhatch] (3,1) rectangle (4,3);
    \fill[pointgreen!10,greenhatch] (1,0) rectangle (3,1);
    \path[dangerrect] (1.45,.49) rectangle (3.44,3.97);
    \drawquartilegrid
    \node[bluesplit,label=above left:$q_1$] (oa1) at (1.34,2.67) {};
    \node[redsplit,label=right:$q_2$] (ox1) at (3.55,2.43) {};
    \node[greensplit,label=below:$q_3$] (oy1) at (2.66,.38) {};
    \draw[avoidarrow,pointblue] (oa1)--+(.32,-.32);
    \draw[avoidarrow,pointred] (ox1)--+(-.32,-.32);
    \draw[avoidarrow,pointgreen] (oy1)--+(.32,.32);
    \node at (2,-.47) {$R_1$};
  \end{scope}
  \begin{scope}[xshift=2.75cm]
    \path[emptycell] (2,1) rectangle (3,2);
    \fill[pointblue!10,bluehatch] (1,2) rectangle (2,3);
    \fill[pointred!10,redhatch] (3,1) rectangle (4,3);
    \fill[pointgreen!10,greenhatch] (1,0) rectangle (3,1);
    \path[dangerrect] (1.45,1.41) rectangle (3.97,3.97);
    \drawquartilegrid
    \node[bluesplit,label=above left:$q_1$] (oa2) at (1.34,2.67) {};
    \node[redsplit,label=right:$q_2$] (ox2) at (3.48,1.30) {};
    \node[greensplit,label=below:$q_3$] (oy2) at (1.67,.36) {};
    \draw[avoidarrow,pointblue] (oa2)--+(.32,-.32);
    \draw[avoidarrow,pointred] (ox2)--+(-.32,-.32);
    \draw[avoidarrow,pointgreen] (oy2)--+(.32,.32);
    \node at (2,-.47) {$R_2$};
  \end{scope}
  \begin{scope}[xshift=5.50cm]
    \path[emptycell] (2,1) rectangle (3,2);
    \fill[pointblue!10,bluehatch] (1,2) rectangle (2,3);
    \fill[pointred!10,redhatch] (3,1) rectangle (4,3);
    \fill[pointgreen!10,greenhatch] (1,0) rectangle (3,1);
    \path[dangerrect] (1.81,.03) rectangle (3.97,2.56);
    \drawquartilegrid
    \node[bluesplit,label=above left:$q_1$] (oa3) at (1.34,2.67) {};
    \node[redsplit,label=right:$q_2$] (ox3) at (3.52,2.74) {};
    \node[greensplit,label=below:$q_3$] (oy3) at (1.70,.43) {};
    \draw[avoidarrow,pointblue] (oa3)--+(.32,-.32);
    \draw[avoidarrow,pointred] (ox3)--+(-.32,-.32);
    \draw[avoidarrow,pointgreen] (oy3)--+(.32,.32);
    \node at (2,-.47) {$R_3$};
  \end{scope}
  \begin{scope}[xshift=8.25cm]
    \path[emptycell] (2,1) rectangle (3,2);
    \fill[pointblue!10,bluehatch] (1,2) rectangle (2,3);
    \fill[pointred!10,redhatch] (3,1) rectangle (4,3);
    \fill[pointgreen!10,greenhatch] (1,0) rectangle (3,1);
    \path[dangerrect] (.03,.03) rectangle (2.57,2.56);
    \drawquartilegrid
    \node[bluesplit,label=above left:$q_1$] (oa4) at (1.34,2.67) {};
    \node[redsplit,label=right:$q_2$] (ox4) at (3.47,2.42) {};
    \node[greensplit,label=below:$q_3$] (oy4) at (2.68,.37) {};
    \draw[avoidarrow,pointblue] (oa4)--+(.32,-.32);
    \draw[avoidarrow,pointred] (ox4)--+(-.32,-.32);
    \draw[avoidarrow,pointgreen] (oy4)--+(.32,.32);
    \node at (2,-.47) {$R_4$};
  \end{scope}
  \begin{scope}[xshift=11.0cm]
    \path[emptycell] (2,1) rectangle (3,2);
    \fill[pointblue!10,bluehatch] (1,2) rectangle (2,3);
    \fill[pointred!10,redhatch] (3,1) rectangle (4,3);
    \fill[pointgreen!10,greenhatch] (1,0) rectangle (3,1);
    \path[dangerrect] (.03,.48) rectangle (3.97,2.56);
    \drawquartilegrid
    \node[bluesplit,label=above left:$q_1$] (oa5) at (1.34,2.67) {};
    \node[redsplit,label=right:$q_2$] (ox5) at (3.50,2.73) {};
    \node[greensplit,label=below:$q_3$] (oy5) at (2.68,.37) {};
    \draw[avoidarrow,pointblue] (oa5)--+(.32,-.32);
    \draw[avoidarrow,pointred] (ox5)--+(-.32,-.32);
    \draw[avoidarrow,pointgreen] (oy5)--+(.32,.32);
    \node at (2,-.47) {$R_5$};
  \end{scope}
\end{tikzpicture}

%% file: figures/strong-net-three-center.tex
\begin{tikzpicture}[x=.792cm,y=.792cm,font=\scriptsize,
  bluesplit/.append style={inner sep=1.9pt},
  redpoint/.append style={inner sep=1.9pt},
  greenpoint/.append style={inner sep=1.9pt},
  avoidarrow/.append style={line width=.62pt},
  every label/.append style={font=\tiny,inner sep=.2pt}]
  \path[emptycell] (2,1) rectangle (3,2);
  \fill[pointblue!10,bluehatch] (1,2) rectangle (2,3);
  \fill[pointred!13,redhatch] (2,2.45) rectangle (3,3);
  \fill[emptycell] (2,2) rectangle (3,2.45);
  \fill[pointgreen!13,greenhatch] (1,1) rectangle (1.74,2);
  \fill[emptycell] (1.74,1) rectangle (2,2);
  \path[dangerrect] (1.85,.03) rectangle (3.97,2.34);
  \drawquartilegrid
  \node[bluesplit,label=above left:$q_1$] (tcF) at (1.30,2.67) {};
  \node[redpoint,label=above right:$q_2$] (tcG) at (2.62,2.45) {};
  \node[greenpoint,label=below left:$q_3$] (tcJ) at (1.74,1.38) {};
  \draw[avoidarrow,pointblue] (tcF)--+(.31,-.31);
  \draw[avoidarrow,pointred] (tcG)--+(0,-.44);
  \draw[avoidarrow,pointgreen] (tcJ)--+(.44,0);
  \node at (2,-.42) {$R$};
\end{tikzpicture}

%% file: figures/strong-net-adjacent-center.tex
\begin{tikzpicture}[
  x=.66cm,y=.66cm,font=\scriptsize,
  bluesplit/.append style={inner sep=1.6pt},
  redpoint/.append style={inner sep=1.6pt},
  greenpoint/.append style={inner sep=1.6pt},
  avoidarrow/.append style={line width=.56pt},
  every label/.append style={font=\tiny,inner sep=.15pt}]
  \begin{scope}
    \path[emptycell] (1,1) rectangle (3,2);
    \fill[pointblue!10,bluehatch] (1,2) rectangle (2,3);
    \fill[pointred!13,redhatch] (2,2.43) rectangle (3,3);
    \fill[emptycell] (2,2) rectangle (3,2.43);
    \fill[pointgreen!13,greenhatch] (1,0) rectangle (3,.54);
    \fill[emptycell] (1,.54) rectangle (3,1);
    \path[dangerrect] (1.44,.03) rectangle (3.97,2.32);
    \drawquartilegrid
    \node[bluesplit,label=above left:$q_1$] (ac1f) at (1.33,2.66) {};
    \node[redpoint,label=above right:$q_2$] (ac1g) at (2.65,2.43) {};
    \node[greenpoint,label=below:$q_3$] (ac1y) at (1.18,.54) {};
    \draw[avoidarrow,pointblue] (ac1f)--+(.32,-.32);
    \draw[avoidarrow,pointred] (ac1g)--+(0,-.42);
    \draw[avoidarrow,pointgreen] (ac1y)--+(0,.42);
    \node at (2,-.44) {$R_1$};
  \end{scope}
  \begin{scope}[xshift=3.72cm]
    \path[emptycell] (1,1) rectangle (3,2);
    \fill[pointblue!10,bluehatch] (1,2) rectangle (2,3);
    \fill[pointred!13,redhatch] (2,2.43) rectangle (3,3);
    \fill[emptycell] (2,2) rectangle (3,2.43);
    \fill[pointgreen!13,greenhatch] (1,0) rectangle (3,.54);
    \fill[emptycell] (1,.54) rectangle (3,1);
    \path[dangerrect] (1.83,.03) rectangle (3.97,2.32);
    \drawquartilegrid
    \node[bluesplit,label=above left:$q_1$] (ac2f) at (1.33,2.66) {};
    \node[redpoint,label=above right:$q_2$] (ac2g) at (2.65,2.43) {};
    \node[greenpoint,label=below:$q_3$] (ac2y) at (1.72,.54) {};
    \draw[avoidarrow,pointblue] (ac2f)--+(.32,-.32);
    \draw[avoidarrow,pointred] (ac2g)--+(0,-.42);
    \draw[avoidarrow,pointgreen] (ac2y)--+(0,.42);
    \node at (2,-.44) {$R_2$};
  \end{scope}
  \begin{scope}[xshift=7.44cm]
    \path[emptycell] (1,1) rectangle (3,2);
    \fill[pointblue!10,bluehatch] (1,2) rectangle (2,3);
    \fill[pointred!13,redhatch] (2,2.43) rectangle (3,3);
    \fill[emptycell] (2,2) rectangle (3,2.43);
    \fill[pointgreen!13,greenhatch] (1,0) rectangle (3,.54);
    \fill[emptycell] (1,.54) rectangle (3,1);
    \path[dangerrect] (.03,.65) rectangle (2.54,2.55);
    \drawquartilegrid
    \node[bluesplit,label=above left:$q_1$] (ac3f) at (1.33,2.66) {};
    \node[redpoint,label=above right:$q_2$] (ac3g) at (2.65,2.43) {};
    \node[greenpoint,label=below:$q_3$] (ac3y) at (1.72,.54) {};
    \draw[avoidarrow,pointblue] (ac3f)--+(.32,-.32);
    \draw[avoidarrow,pointred] (ac3g)--+(0,-.42);
    \draw[avoidarrow,pointgreen] (ac3y)--+(0,.42);
    \node at (2,-.44) {$R_3$};
  \end{scope}
  \begin{scope}[xshift=11.16cm]
    \path[emptycell] (1,1) rectangle (3,2);
    \fill[pointblue!10,bluehatch] (1,2) rectangle (2,3);
    \fill[pointred!13,redhatch] (2,2.43) rectangle (3,3);
    \fill[emptycell] (2,2) rectangle (3,2.43);
    \fill[pointgreen!13,greenhatch] (1,0) rectangle (3,.54);
    \fill[emptycell] (1,.54) rectangle (3,1);
    \path[dangerrect] (.03,.65) rectangle (3.97,2.32);
    \drawquartilegrid
    \node[bluesplit,label=above left:$q_1$] (ac4f) at (1.33,2.66) {};
    \node[redpoint,label=above right:$q_2$] (ac4g) at (2.65,2.43) {};
    \node[greenpoint,label=below:$q_3$] (ac4y) at (1.72,.54) {};
    \draw[avoidarrow,pointblue] (ac4f)--+(.32,-.32);
    \draw[avoidarrow,pointred] (ac4g)--+(0,-.42);
    \draw[avoidarrow,pointgreen] (ac4y)--+(0,.42);
    \node at (2,-.44) {$R_4$};
  \end{scope}
\end{tikzpicture}

%% file: figures/strong-net-se-anchor.tex
\begin{tikzpicture}[
  x=.792cm,y=.792cm,font=\scriptsize,
  bluepoint/.append style={inner sep=1.9pt},
  redpoint/.append style={inner sep=1.9pt},
  greenpoint/.append style={inner sep=1.9pt},
  every label/.append style={font=\tiny,inner sep=.2pt}]
  \begin{scope}
    \path[emptycell] (1,1) rectangle (3,2);
    \path[emptycell] (2,2) rectangle (3,3);
    \fill[pointblue!9,bluehatch] (1,2) rectangle (2,3);
    \fill[emptycell] (1.78,2) rectangle (2,3);
    \fill[emptycell] (1,2) rectangle (2,2.20);
    \drawquartilegrid
    \draw[guide,pointblue] (1.78,2)--(1.78,3);
    \draw[guide,pointred] (1,2.20)--(2,2.20);
    \node[bluepoint,label=above left:$q_1$]
      at (1.78,2.72) {};
    \node[redpoint,label=above right:$q_2$]
      at (1.30,2.20) {};
    \node[align=center] at (2,-.48)
      {(a) $q_1=\East(F)$ and $q_2=\South(F)$};
  \end{scope}
  \begin{scope}[xshift=5.25cm]
    \path[emptycell] (1,1) rectangle (3,2);
    \path[emptycell] (2,2) rectangle (3,3);
    \fill[pointblue!12,bluehatch] (3.12,1) rectangle (4,2.88);
    \fill[pointred!12,redhatch] (1.08,0) rectangle (4,1.26);
    \drawquartilegrid
    \node[greenpoint,label=above left:$q_3$] at (3.05,1.38) {};
    \node[text=pointblue!75!black,font=\tiny] at (3.57,2.43) {$Z_E$};
    \node[text=pointred!78!black,font=\tiny] at (2.30,.47) {$Z_S$};
    \node at (2,-.42) {(b) the southeast point};
  \end{scope}
\end{tikzpicture}

%% file: figures/strong-net-one-center-easy.tex
\begin{tikzpicture}[
  x=.792cm,y=.792cm,font=\scriptsize,
  greenpoint/.append style={inner sep=1.9pt},
  every label/.append style={font=\tiny,inner sep=.2pt}]
  \path[emptycell] (1,1) rectangle (3,2);
  \path[emptycell] (2,2) rectangle (3,3);
  \fill[pointblue!13,bluehatch] (3.18,1) rectangle (4,3);
  \fill[emptycell] (3,1) rectangle (3.18,3);
  \fill[pointred!11,redhatch] (1,0) rectangle (4,1);
  \drawquartilegrid
  \draw[guide,pointgreen] (3.18,1)--(3.18,3);
  \node[greenpoint,label=below left:$q_3$]
    at (3.18,1.42) {};
  \node[text=pointblue!75!black,font=\tiny] at (3.7,2.35) {$Z_E$};
  \node[text=pointred!78!black,font=\tiny] at (2.6,.33) {$Z_S$};
\end{tikzpicture}

%% file: figures/strong-net-one-center-rich.tex
\begin{tikzpicture}[
  x=.66cm,y=.66cm,font=\scriptsize,
  bluesplit/.append style={inner sep=1.6pt},
  redpoint/.append style={inner sep=1.6pt},
  greenpoint/.append style={inner sep=1.6pt},
  avoidarrow/.append style={line width=.56pt},
  every label/.append style={font=\tiny,inner sep=.15pt}]
  \begin{scope}
    \path[emptycell] (1,1) rectangle (3,2);
    \path[emptycell] (2,2) rectangle (3,3);
    \fill[pointblue!10,bluehatch] (1,2) rectangle (2,3);
    \fill[pointred!12,redhatch] (3.32,1) rectangle (4,3);
    \fill[emptycell] (3,1) rectangle (3.32,3);
    \fill[pointgreen!12,greenhatch] (1,0) rectangle (3,.52);
    \fill[emptycell] (1,.52) rectangle (3,1);
    \fill[pointred!20,opacity=.65,redhatch] (3.32,0) rectangle (4,1);
    \fill[pointgreen!20,opacity=.65,greenhatch] (3,0) rectangle (4,.52);
    \path[dangerrect] (1.44,.03) rectangle (3.21,3.97);
    \drawquartilegrid
    \draw[guide,pointred] (3.32,0)--(3.32,1);
    \draw[guide,pointgreen] (3,.52)--(4,.52);
    \node[bluesplit,label=above left:$q_1$] (or1f) at (1.33,2.67) {};
    \node[redpoint,label=right:$q_2$] (or1u) at (3.32,2.84) {};
    \node[greenpoint,label=below:$q_3$] (or1v) at (1.18,.52) {};
    \draw[avoidarrow,pointblue] (or1f)--+(.38,0);
    \draw[avoidarrow,pointred] (or1u)--+(-.42,0);
    \draw[avoidarrow,pointgreen] (or1v)--+(0,.42);
    \node[font=\scriptsize] at (3.68,.76) {$\rho$};
    \node[font=\scriptsize] at (3.14,.24) {$\eta$};
    \node at (2,-.44) {$R_1$};
  \end{scope}
  \begin{scope}[xshift=3.72cm]
    \path[emptycell] (1,1) rectangle (3,2);
    \path[emptycell] (2,2) rectangle (3,3);
    \fill[pointblue!10,bluehatch] (1,2) rectangle (2,3);
    \fill[pointred!12,redhatch] (3.32,1) rectangle (4,3);
    \fill[emptycell] (3,1) rectangle (3.32,3);
    \fill[pointgreen!12,greenhatch] (1,0) rectangle (3,.52);
    \fill[emptycell] (1,.52) rectangle (3,1);
    \fill[pointred!20,opacity=.65,redhatch] (3.32,0) rectangle (4,1);
    \fill[pointgreen!20,opacity=.65,greenhatch] (3,0) rectangle (4,.52);
    \path[dangerrect] (1.44,.03) rectangle (3.97,2.73);
    \drawquartilegrid
    \draw[guide,pointred] (3.32,0)--(3.32,1);
    \draw[guide,pointgreen] (3,.52)--(4,.52);
    \node[bluesplit,label=above left:$q_1$] (or2f) at (1.33,2.67) {};
    \node[redpoint,label=right:$q_2$] (or2u) at (3.32,2.84) {};
    \node[greenpoint,label=below:$q_3$] (or2v) at (1.18,.52) {};
    \draw[avoidarrow,pointblue] (or2f)--+(.38,0);
    \draw[avoidarrow,pointred] (or2u)--+(0,-.42);
    \draw[avoidarrow,pointgreen] (or2v)--+(0,.42);
    \node at (2,-.44) {$R_2$};
  \end{scope}
  \begin{scope}[xshift=7.44cm]
    \path[emptycell] (1,1) rectangle (3,2);
    \path[emptycell] (2,2) rectangle (3,3);
    \fill[pointblue!10,bluehatch] (1,2) rectangle (2,3);
    \fill[pointred!12,redhatch] (3.32,1) rectangle (4,3);
    \fill[emptycell] (3,1) rectangle (3.32,3);
    \fill[pointgreen!12,greenhatch] (1,0) rectangle (3,.52);
    \fill[emptycell] (1,.52) rectangle (3,1);
    \fill[pointred!20,opacity=.65,redhatch] (3.32,0) rectangle (4,1);
    \fill[pointgreen!20,opacity=.65,greenhatch] (3,0) rectangle (4,.52);
    \path[dangerrect] (.03,.63) rectangle (3.97,2.56);
    \drawquartilegrid
    \draw[guide,pointred] (3.32,0)--(3.32,1);
    \draw[guide,pointgreen] (3,.52)--(4,.52);
    \node[bluesplit,label=above left:$q_1$] (or3f) at (1.33,2.67) {};
    \node[redpoint,label=right:$q_2$] (or3u) at (3.32,2.84) {};
    \node[greenpoint,label=below:$q_3$] (or3v) at (1.18,.52) {};
    \draw[avoidarrow,pointblue] (or3f)--+(0,-.42);
    \draw[avoidarrow,pointred] (or3u)--+(-.42,0);
    \draw[avoidarrow,pointgreen] (or3v)--+(0,.42);
    \node at (2,-.44) {$R_3$};
  \end{scope}
  \begin{scope}[xshift=11.16cm]
    \path[emptycell] (1,1) rectangle (3,2);
    \path[emptycell] (2,2) rectangle (3,3);
    \fill[pointblue!10,bluehatch] (1,2) rectangle (2,3);
    \fill[pointred!12,redhatch] (3.32,1) rectangle (4,3);
    \fill[emptycell] (3,1) rectangle (3.32,3);
    \fill[pointgreen!12,greenhatch] (1,0) rectangle (3,.52);
    \fill[emptycell] (1,.52) rectangle (3,1);
    \fill[pointred!20,opacity=.65,redhatch] (3.32,0) rectangle (4,1);
    \fill[pointgreen!20,opacity=.65,greenhatch] (3,0) rectangle (4,.52);
    \path[dangerrect] (1.29,.03) rectangle (3.97,2.56);
    \drawquartilegrid
    \draw[guide,pointred] (3.32,0)--(3.32,1);
    \draw[guide,pointgreen] (3,.52)--(4,.52);
    \node[bluesplit,label=above left:$q_1$] (or4f) at (1.33,2.67) {};
    \node[redpoint,label=right:$q_2$] (or4u) at (3.32,2.84) {};
    \node[greenpoint,label=below:$q_3$] (or4v) at (1.18,.52) {};
    \draw[avoidarrow,pointblue] (or4f)--+(0,-.42);
    \draw[avoidarrow,pointred] (or4u)--+(-.42,0);
    \draw[avoidarrow,pointgreen] (or4v)--+(.42,0);
    \node at (2,-.44) {$R_4$};
  \end{scope}
\end{tikzpicture}

%% file: figures/strong-net-one-center-poor.tex
\begin{tikzpicture}[
  x=.792cm,y=.792cm,font=\scriptsize,
  bluepoint/.append style={inner sep=1.85pt},
  redpoint/.append style={inner sep=1.85pt},
  greenpoint/.append style={inner sep=1.85pt},
  every label/.append style={font=\tiny,inner sep=.2pt}]
  \path[emptycell] (1,1) rectangle (3,2);
  \path[emptycell] (2,2) rectangle (3,3);
  \fill[pointred!13,redhatch] (3.30,1) rectangle (4,3);
  \fill[emptycell] (3,1) rectangle (3.30,3);
  \fill[pointblue!12,bluehatch] (1,0) rectangle (3,.55);
  \fill[emptycell] (1,.55) rectangle (3,1);
  \path[focuszone] (3,.55) rectangle (3.30,1);
  \drawquartilegrid
  \draw[guide,pointred] (3.30,0)--(3.30,3);
  \draw[guide,pointblue] (1,.55)--(4,.55);
  \node[greenpoint,label=left:$q_3$] at (3.12,.82) {};
  \node[redpoint,label=above left:] at (3.30,2.45) {};
  \node[bluepoint,label=above:] at (1.43,.55) {};
  \node[text=pointred!78!black,font=\tiny] at (3.64,2.25) {$X$};
  \node[text=pointblue!75!black,font=\tiny] at (2.38,.28) {$Y$};
\end{tikzpicture}

%% file: figures/strong-net-lower-bound-k2.tex
\begin{tikzpicture}[
    x = 0.35cm,
    y = 0.35cm,
    frame/.style = {draw = black!45, line width = 0.4pt},
    diagonal/.style = {densely dashed, draw = black!30, line width = 0.4pt},
    point/.style = {circle, fill = black, inner sep = 1.1pt}
]
    \draw[frame] (0.5, 0.5) rectangle (14.5, 14.5);
    \draw[diagonal] (0.5, 0.5) -- (14.5, 14.5);

    \foreach \i/\j in {1/7, 2/8, 3/9, 4/12, 5/13, 6/14, 7/1, 8/2, 9/3, 10/11, 11/10, 12/4, 13/5, 14/6} {
        \node[point] at (\i, \j) {};
    }
\end{tikzpicture}

%% file: figures/strong-net-lower-bound-k3.tex
\begin{tikzpicture}[
    x = 0.19cm,
    y = 0.19cm,
    frame/.style = {draw = black!45, line width = 0.4pt},
    diagonal/.style = {densely dashed, draw = black!30, line width = 0.4pt},
    point/.style = {circle, fill = black, inner sep = 1.1pt}
]
    \draw[frame] (0.5, 0.5) rectangle (26.5, 26.5);
    \draw[diagonal] (0.5, 0.5) -- (26.5, 26.5);

    \foreach \i/\j in {1/12, 2/13, 3/16, 4/9, 5/7, 6/19, 7/5, 8/21, 9/22, 10/23, 11/24, 12/1, 13/2, 14/25, 15/26, 16/3, 17/4, 18/20, 19/6, 20/18, 21/8, 22/17, 23/10, 24/11, 25/14, 26/15} {
        \node[point] at (\i, \j) {};
    }
\end{tikzpicture}

%% file: figures/strong-net-lower-bound-k4.tex
\begin{tikzpicture}[
    x = 0.31cm,
    y = 0.31cm,
    frame/.style = {draw = black!45, line width = 0.4pt},
    diagonal/.style = {densely dashed, draw = black!30, line width = 0.4pt},
    point/.style = {circle, fill = black, inner sep = 1.1pt}
]
    \draw[frame] (0.5, 0.5) rectangle (16.5, 16.5);
    \draw[diagonal] (0.5, 0.5) -- (16.5, 16.5);

    \foreach \i/\j in {1/9, 2/10, 3/11, 4/5, 5/4, 6/14, 7/15, 8/16, 9/1, 10/2, 11/3, 12/13, 13/12, 14/6, 15/7, 16/8} {
        \node[point] at (\i, \j) {};
    }
\end{tikzpicture}

%% file: figures/giveaway.tex
\begin{tikzpicture}[
    x = 1cm,
    y = 1cm,
    axis/.style = {->, draw = black!45, line width = 0.35pt},
    bisector/.style = {draw = black!60, line width = 0.55pt},
    hull/.style = {densely dotted, draw = black, line width = 0.85pt},
    support/.style = {densely dashed, draw = black!60, line width = 0.65pt},
    separator/.style = {draw = black, line width = 1pt},
    vector/.style = {-{Latex[length = 2mm]}, draw = black, line width = 0.8pt},
    point/.style = {circle, fill = black, inner sep = 1.2pt},
    every node/.style = {font = \scriptsize}
]
    \begin{scope}
        \begin{scope}
            \clip (-2.6, -1.35) rectangle (1.6, 3);
            \fill[black!9]
                (-2.6, -1.35) --
                (1.6, -1.35) --
                (1.6, 1.5) --
                (0, 1.5) --
                (-1, 0.5) --
                (-2.6, 0.5) -- cycle;
            \draw[bisector]
                (-2.6, 0.5) --
                (-1, 0.5) --
                (0, 1.5) --
                (1.6, 1.5);
            \draw[hull] (-2.6, 1.5) -- (1.6, 1.5);
            \draw[support] (-2.5, 3) -- (-0.325, -1.35);
            \draw[support] (-0.5, 3) -- (1.6, -1.2);
            \draw[separator] (-2.6, 2) -- (1.6, 2);
        \end{scope}

        \draw[axis] (-2.55, 0) -- (1.55, 0);
        \draw[axis] (0, -1.3) -- (0, 2.95);

        \draw[line width = 0.85pt]
            (0, 1) --
            (1, 0) --
            (0, -1) --
            (-1, 0) -- cycle;
        \node at (0.52, -0.52) {$B$};

        \draw[vector] (0, 0) -- (1, 0) node[right] {$u$};
        \draw[vector] (0, 0) -- (-1, 2);
        \node[point, label = {[label distance = 1pt]left:$-u$}] at (-1, 0) {};
        \node[point] at (1, 0) {};
        \node[point, label = {[label distance = 2pt]above left:$b$}] at (-1, 2) {};
        \node[point, label = {[label distance = 1pt]below right:$0$}] at (0, 0) {};

        \node[font = \small] at (-0.5, -1.7) {$\ell_1$};
    \end{scope}

    \begin{scope}[xshift = 5.2cm]
        \begin{scope}
            \clip (-2.6, -1.35) rectangle (1.6, 3);
            \fill[black!9]
                (-2.6, -1.35) --
                (1.6, -1.35) --
                (1.6, 2.5023) --
                (1.4, 2.3275) --
                (1.2, 2.1496) --
                (1, 1.9680) --
                (0.8, 1.7823) --
                (0.6, 1.5937) --
                (0.4, 1.4087) --
                (0.2, 1.2444) --
                (0, 1.1231) --
                (-0.2, 1.0509) --
                (-0.4, 1.0130) --
                (-0.6, 0.9870) --
                (-0.8, 0.9491) --
                (-1, 0.8769) --
                (-1.2, 0.7556) --
                (-1.4, 0.5913) --
                (-1.6, 0.4063) --
                (-1.8, 0.2177) --
                (-2, 0.0320) --
                (-2.2, -0.1496) --
                (-2.4, -0.3275) --
                (-2.6, -0.5023) -- cycle;
            \draw[bisector]
                (-2.6, -0.5023) --
                (-2.4, -0.3275) --
                (-2.2, -0.1496) --
                (-2, 0.0320) --
                (-1.8, 0.2177) --
                (-1.6, 0.4063) --
                (-1.4, 0.5913) --
                (-1.2, 0.7556) --
                (-1, 0.8769) --
                (-0.8, 0.9491) --
                (-0.6, 0.9870) --
                (-0.4, 1.0130) --
                (-0.2, 1.0509) --
                (0, 1.1231) --
                (0.2, 1.2444) --
                (0.4, 1.4087) --
                (0.6, 1.5937) --
                (0.8, 1.7823) --
                (1, 1.9680) --
                (1.2, 2.1496) --
                (1.4, 2.3275) --
                (1.6, 2.5023);
            \draw[hull] (-2.6, -0.3536) -- (1.6, 2.9799);
            \draw[support] (-2.6, 2.6302) -- (-0.6099, -1.35);
            \draw[support] (-0.2151, 3) -- (1.6, -0.6302);
            \draw[separator] (-2.6, 0.7301) -- (0.260, 3);
        \end{scope}

        \draw[axis] (-2.55, 0) -- (1.55, 0);
        \draw[axis] (0, -1.3) -- (0, 2.95);

        \draw[line width = 0.85pt, domain = 0:360, samples = 140, smooth, variable = \t]
            plot (
                {cos(\t) / (abs(cos(\t))^4 + abs(sin(\t))^4)^0.25},
                {sin(\t) / (abs(cos(\t))^4 + abs(sin(\t))^4)^0.25}
            );
        \node at (0.72, -0.7) {$B$};

        \draw[vector] (0, 0) -- (0.9198, 0.7301) node[right] {$u$};
        \draw[vector] (0, 0) -- (-1, 2);
        \node[point, label = {[label distance = 1pt]left:$-u$}] at (-0.9198, -0.7301) {};
        \node[point] at (0.9198, 0.7301) {};
        \node[point, label = {[label distance = 2pt]above left:$b$}] at (-1, 2) {};
        \node[point, label = {[label distance = 1pt]below right:$0$}] at (0, 0) {};

        \node[font = \small] at (-0.5, -1.7) {$\ell_4$};
    \end{scope}

    \begin{scope}[xshift = 10.4cm]
        \begin{scope}
            \clip (-2.6, -1.35) rectangle (1.6, 3);
            \fill[black!9]
                (-2.6, -1.35) --
                (1.6, -1.35) --
                (1.6, 2.6) --
                (0, 1) --
                (-1, 1) --
                (-2.6, -0.6) -- cycle;
            \draw[bisector]
                (-2.6, -0.6) --
                (-1, 1) --
                (0, 1) --
                (1.6, 2.6);
            \draw[hull] (-2.6, -0.6) -- (1, 3);
            \draw[support] (-2.6, 2.2) -- (-0.825, -1.35);
            \draw[support] (0, 3) -- (1.6, -0.2);
            \draw[separator] (-2.6, 0.4) -- (0, 3);
        \end{scope}

        \draw[axis] (-2.55, 0) -- (1.55, 0);
        \draw[axis] (0, -1.3) -- (0, 2.95);

        \draw[line width = 0.85pt] (-1, -1) rectangle (1, 1);
        \node at (0.72, -0.7) {$B$};

        \draw[vector] (0, 0) -- (1, 1) node[right] {$u$};
        \draw[vector] (0, 0) -- (-1, 2);
        \node[point, label = {[label distance = 1pt]left:$-u$}] at (-1, -1) {};
        \node[point] at (1, 1) {};
        \node[point, label = {[label distance = 2pt]above left:$b$}] at (-1, 2) {};
        \node[point, label = {[label distance = 1pt]below right:$0$}] at (0, 0) {};

        \node[font = \small] at (-0.5, -1.7) {$\ell_\infty$};
    \end{scope}
\end{tikzpicture}

%% file: figures/rounding-counterexample.tex
\begin{tikzpicture}[
    x = 0.4cm,
    y = 0.4cm,
    axis/.style = {->, draw = black!45, line width = 0.35pt},
    bisector/.style = {draw = black!60, line width = 0.65pt},
    reference/.style = {densely dotted, draw = black, line width = 0.7pt},
    voter/.style = {circle, fill = black, inner sep = 1.2pt},
    candidate/.style = {rectangle, fill = black, inner sep = 1.5pt},
    netpoint/.style = {rectangle, draw = black, fill = white, line width = 0.7pt, inner sep = 0.8pt, rotate = 45},
    every node/.style = {font = \scriptsize}
]
    \begin{scope}
        \begin{scope}
            \clip (0, 0) rectangle (10, 10);
            \fill[black!9]
                (2.3073, 0) --
                (2.4113, 0.5) --
                (2.5189, 1) --
                (2.6308, 1.5) --
                (2.7475, 2) --
                (2.8702, 2.5) --
                (3, 3) --
                (3.1389, 3.5) --
                (3.2898, 4) --
                (3.4581, 4.5) --
                (3.6603, 5) --
                (3.9647, 5.5) --
                (4.2997, 6) --
                (4.6473, 6.5) --
                (5, 7) --
                (5.3527, 7.5) --
                (5.7003, 8) --
                (6.0353, 8.5) --
                (6.3397, 9) --
                (6.5419, 9.5) --
                (6.7102, 10) --
                (10, 10) --
                (10, 0) -- cycle;
            \draw[reference] (1.5, 0) -- (6.5, 10);
            \draw[bisector]
                (2.3073, 0) --
                (2.4113, 0.5) --
                (2.5189, 1) --
                (2.6308, 1.5) --
                (2.7475, 2) --
                (2.8702, 2.5) --
                (3, 3) --
                (3.1389, 3.5) --
                (3.2898, 4) --
                (3.4581, 4.5) --
                (3.6603, 5) --
                (3.9647, 5.5) --
                (4.2997, 6) --
                (4.6473, 6.5) --
                (5, 7) --
                (5.3527, 7.5) --
                (5.7003, 8) --
                (6.0353, 8.5) --
                (6.3397, 9) --
                (6.5419, 9.5) --
                (6.7102, 10);
        \end{scope}

        \draw[axis] (0, 0) -- (10.25, 0);
        \draw[axis] (0, 0) -- (0, 10.25);

        \node[voter] at (2, 1) {};
        \node[below right] at (2, 1) {$v_1$};
        \node[netpoint] at (4, 5) {};
        \node[above left] at (4, 5) {$q$};
        \node[voter] at (6, 9) {};
        \node[below right] at (6, 9) {$v_2$};
        \node[candidate] at (9, 5) {};
        \node[right] at (9, 5) {$a$};
        \node[candidate] at (1, 9) {};
        \node[above left] at (1, 9) {$b$};

        \node[font = \small] at (5, -0.8) {$\ell_{2 - \delta}$};
    \end{scope}

    \begin{scope}[xshift = 4.6cm]
        \begin{scope}
            \clip (0, 0) rectangle (10, 10);
            \fill[black!9] (1.5, 0) -- (6.5, 10) -- (10, 10) -- (10, 0) -- cycle;
            \draw[bisector] (1.5, 0) -- (6.5, 10) node[pos = 0.35, left] {$L$};
        \end{scope}

        \draw[axis] (0, 0) -- (10.25, 0);
        \draw[axis] (0, 0) -- (0, 10.25);

        \node[voter] at (2, 1) {};
        \node[below right] at (2, 1) {$v_1$};
        \node[netpoint] at (4, 5) {};
        \node[above left] at (4, 5) {$q$};
        \node[voter] at (6, 9) {};
        \node[below right] at (6, 9) {$v_2$};
        \node[candidate] at (9, 5) {};
        \node[right] at (9, 5) {$a$};
        \node[candidate] at (1, 9) {};
        \node[above left] at (1, 9) {$b$};

        \node[font = \small] at (5, -0.8) {$\ell_2$};
    \end{scope}

    \begin{scope}[xshift = 9.2cm]
        \begin{scope}
            \clip (0, 0) rectangle (10, 10);
            \fill[black!9]
                (0.8935, 0) --
                (1.2406, 0.5) --
                (1.5935, 1) --
                (1.9492, 1.5) --
                (2.3046, 2) --
                (2.6561, 2.5) --
                (3, 3) --
                (3.3318, 3.5) --
                (3.6469, 4) --
                (3.9398, 4.5) --
                (4.2039, 5) --
                (4.4329, 5.5) --
                (4.6357, 6) --
                (4.8219, 6.5) --
                (5, 7) --
                (5.1781, 7.5) --
                (5.3643, 8) --
                (5.5671, 8.5) --
                (5.7961, 9) --
                (6.0602, 9.5) --
                (6.3531, 10) --
                (10, 10) --
                (10, 0) -- cycle;
            \draw[reference] (1.5, 0) -- (6.5, 10);
            \draw[bisector]
                (0.8935, 0) --
                (1.2406, 0.5) --
                (1.5935, 1) --
                (1.9492, 1.5) --
                (2.3046, 2) --
                (2.6561, 2.5) --
                (3, 3) --
                (3.3318, 3.5) --
                (3.6469, 4) --
                (3.9398, 4.5) --
                (4.2039, 5) --
                (4.4329, 5.5) --
                (4.6357, 6) --
                (4.8219, 6.5) --
                (5, 7) --
                (5.1781, 7.5) --
                (5.3643, 8) --
                (5.5671, 8.5) --
                (5.7961, 9) --
                (6.0602, 9.5) --
                (6.3531, 10);
        \end{scope}

        \draw[axis] (0, 0) -- (10.25, 0);
        \draw[axis] (0, 0) -- (0, 10.25);

        \node[voter] at (2, 1) {};
        \node[below right] at (2, 1) {$v_1$};
        \node[netpoint] at (4, 5) {};
        \node[above left] at (4, 5) {$q$};
        \node[voter] at (6, 9) {};
        \node[below right] at (6, 9) {$v_2$};
        \node[candidate] at (9, 5) {};
        \node[right] at (9, 5) {$a$};
        \node[candidate] at (1, 9) {};
        \node[above left] at (1, 9) {$b$};

        \node[font = \small] at (5, -0.8) {$\ell_{2 + \delta}$};
    \end{scope}
\end{tikzpicture}

%% file: bibliography.bib
@book{Condorcet85,
  title = {Essai sur l'application de l'analyse {\`a} la probabilit{\'e} des decisions rendues {\`a} la pluralit{\'e} des voix.},
  author = {Condorcet, Nicolas de},
  year = {1785},
  publisher = {Imprimerie Royale}
}

@article{Hotelling29,
  issn = {00130133, 14680297},
  url = {http://www.jstor.org/stable/2224214},
  author = {Harold Hotelling},
  journal = {The Economic Journal},
  number = {153},
  pages = {41--57},
  publisher = {[Royal Economic Society, Wiley]},
  title = {Stability in Competition},
  urldate = {2025-12-15},
  volume = {39},
  year = {1929}
}

@article{Black48,
  title = {On the Rationale of Group Decision-Making},
  author = {Black, Duncan},
  journal = {J. Polit. Econ.},
  publisher = {University of Chicago Press},
  volume = {56},
  number = {1},
  pages = {23--34},
  month = feb,
  year = {1948}
}

@article{Simpson69,
  title = {On Defining Areas of Voter Choice: Professor Tullock on Stable Voting},
  author = {Simpson, Paul B.},
  journal = {The Quarterly Journal of Economics},
  volume = {83},
  number = {3},
  pages = {478--490},
  year = {1969},
  month = aug,
  doi = {10.2307/1880533},
  url = {https://doi.org/10.2307/1880533}
}

@article{HausslerW87,
  title = {Epsilon-Nets and Simplex Range Queries},
  author = {Haussler, David and Welzl, Emo},
  journal = {Discrete \& Computational Geometry},
  volume = {2},
  number = {1},
  pages = {127--151},
  year = {1987},
  doi = {10.1007/BF02187876},
  url = {https://doi.org/10.1007/BF02187876}
}

@article{Durier89,
  title = {Continuous Location Theory Under Majority Rule},
  author = {Durier, Roland},
  journal = {Mathematics of Operations Research},
  volume = {14},
  number = {2},
  pages = {258--274},
  year = {1989},
  month = may,
  doi = {10.1287/moor.14.2.258},
  url = {https://doi.org/10.1287/moor.14.2.258}
}

@article{AlonBFK92,
  title = {Point Selections and Weak $\varepsilon$-Nets for Convex Hulls},
  author = {Alon, Noga and B{\'a}r{\'a}ny, Imre and F{\"u}redi, Zolt{\'a}n and Kleitman, Daniel J.},
  journal = {Combinatorics, Probability and Computing},
  volume = {1},
  number = {3},
  pages = {189--200},
  year = {1992},
  doi = {10.1017/S0963548300000225},
  url = {https://doi.org/10.1017/S0963548300000225}
}

@article{Carrizosa96,
  title = {A Characterization of Halfspace Depth},
  author = {Carrizosa, Emilio},
  journal = {Journal of Multivariate Analysis},
  volume = {58},
  number = {1},
  pages = {21--26},
  year = {1996},
  month = jul,
  publisher = {Elsevier},
  doi = {10.1006/jmva.1996.0037},
  url = {https://doi.org/10.1006/jmva.1996.0037}
}

@article{AhnCCGO04,
  title = {Competitive Facility Location: The {Voronoi} Game},
  author = {Ahn, Hee-Kap and Cheng, Siu-Wing and Cheong, Otfried and Golin, Mordecai and van Oostrum, Ren{\'e}},
  journal = {Theoretical Computer Science},
  volume = {310},
  number = {1--3},
  pages = {457--467},
  year = {2004},
  publisher = {Elsevier},
  doi = {10.1016/j.tcs.2003.09.004},
  url = {https://doi.org/10.1016/j.tcs.2003.09.004}
}

@article{CheongHLM04,
  title = {The One-Round {Voronoi} Game},
  author = {Cheong, Otfried and Har-Peled, Sariel and Linial, Nathan and Matou{\v{s}}ek, Ji{\v{r}}{\'i}},
  journal = {Discrete \& Computational Geometry},
  volume = {31},
  number = {1},
  pages = {125--138},
  year = {2004},
  publisher = {Springer},
  doi = {10.1007/s00454-003-2951-4},
  url = {https://doi.org/10.1007/s00454-003-2951-4}
}

@article{ChawlaRRS06,
  title = {Min--Max Payoffs in a Two-Player Location Game},
  author = {Chawla, Shuchi and Rajan, Uday and Ravi, R. and Sinha, Amitabh},
  journal = {Operations Research Letters},
  volume = {34},
  number = {5},
  pages = {499--507},
  year = {2006},
  publisher = {Elsevier},
  doi = {10.1016/j.orl.2005.10.002},
  url = {https://doi.org/10.1016/j.orl.2005.10.002}
}

@incollection{AgarwalPS08,
  title = {State of the Union (of Geometric Objects)},
  author = {Agarwal, Pankaj K. and Pach, J{\'a}nos and Sharir, Micha},
  booktitle = {Surveys on Discrete and Computational Geometry: Twenty Years Later},
  editor = {Goodman, Jacob E. and Pach, J{\'a}nos and Pollack, Richard},
  series = {Contemporary Mathematics},
  volume = {453},
  pages = {9--48},
  publisher = {American Mathematical Society},
  address = {Providence, RI},
  year = {2008},
  doi = {10.1090/conm/453/08794},
  url = {https://doi.org/10.1090/conm/453/08794}
}

@article{AronovAHLRSS09,
  title = {Small Weak Epsilon-Nets},
  author = {Aronov, Boris and Aurenhammer, Franz and Hurtado, Ferran and Langerman, Stefan and Rappaport, David and Seara, Carlos and Smorodinsky, Shakhar},
  journal = {Computational Geometry},
  volume = {42},
  number = {5},
  pages = {455--462},
  year = {2009},
  month = jul,
  publisher = {Elsevier},
  doi = {10.1016/j.comgeo.2008.02.005},
  url = {https://doi.org/10.1016/j.comgeo.2008.02.005}
}

@article{MustafaR09,
  title = {An Optimal Extension of the Centerpoint Theorem},
  author = {Mustafa, Nabil H. and Ray, Saurabh},
  journal = {Computational Geometry},
  volume = {42},
  number = {6--7},
  pages = {505--510},
  year = {2009},
  month = aug,
  publisher = {Elsevier},
  doi = {10.1016/j.comgeo.2007.10.004},
  url = {https://doi.org/10.1016/j.comgeo.2007.10.004}
}

@article{AronovES10,
  title = {Small-Size $\varepsilon$-Nets for Axis-Parallel Rectangles and Boxes},
  author = {Aronov, Boris and Ezra, Esther and Sharir, Micha},
  journal = {SIAM Journal on Computing},
  volume = {39},
  number = {7},
  pages = {3248--3282},
  year = {2010},
  publisher = {SIAM},
  doi = {10.1137/090762968},
  url = {https://doi.org/10.1137/090762968}
}

@inproceedings{AlonFPT11,
  title = {Sum of Us: Strategyproof Selection from the Selectors},
  author = {Alon, Noga and Fischer, Felix A. and Procaccia, Ariel D. and Tennenholtz, Moshe},
  booktitle = {Proceedings of the 13th Conference on Theoretical Aspects of Rationality and Knowledge},
  editor = {Apt, Krzysztof R.},
  pages = {101--110},
  publisher = {ACM},
  year = {2011},
  doi = {10.1145/2000378.2000390},
  url = {https://doi.org/10.1145/2000378.2000390}
}

@inproceedings{GoelL12,
  title = {Triadic Consensus: A Randomized Algorithm for Voting in a Crowd},
  author = {Goel, Ashish and Lee, David},
  booktitle = {Internet and Network Economics},
  editor = {Goldberg, Paul W.},
  series = {Lecture Notes in Computer Science},
  volume = {7695},
  pages = {434--447},
  publisher = {Springer},
  address = {Berlin, Heidelberg},
  year = {2012},
  doi = {10.1007/978-3-642-35311-6_32},
  url = {https://doi.org/10.1007/978-3-642-35311-6_32}
}

@article{BanikBD13,
  title = {Optimal Strategies for the One-Round Discrete {Voronoi} Game on a Line},
  author = {Banik, Aritra and Bhattacharya, Bhaswar B. and Das, Sandip},
  journal = {Journal of Combinatorial Optimization},
  volume = {26},
  number = {4},
  pages = {655--669},
  year = {2013},
  month = nov,
  publisher = {Springer},
  doi = {10.1007/s10878-011-9447-6},
  url = {https://doi.org/10.1007/s10878-011-9447-6}
}

@article{HolzmanM13,
  title = {Impartial Nominations for a Prize},
  author = {Holzman, Ron and Moulin, Herv{\'e}},
  journal = {Econometrica},
  volume = {81},
  number = {1},
  pages = {173--196},
  year = {2013},
  month = jan,
  doi = {10.3982/ECTA10523},
  url = {https://doi.org/10.3982/ECTA10523}
}

@article{PachT13,
  title = {Tight Lower Bounds for the Size of Epsilon-Nets},
  author = {Pach, J{\'a}nos and Tardos, G{\'a}bor},
  journal = {Journal of the American Mathematical Society},
  volume = {26},
  number = {3},
  pages = {645--658},
  year = {2013},
  month = jul,
  publisher = {American Mathematical Society},
  doi = {10.1090/S0894-0347-2012-00759-0},
  url = {https://doi.org/10.1090/S0894-0347-2012-00759-0}
}

@article{AshokAG14,
  title = {Small Strong Epsilon Nets},
  author = {Ashok, Pradeesha and Azmi, Umair and Govindarajan, Sathish},
  journal = {Computational Geometry},
  volume = {47},
  number = {9},
  pages = {899--909},
  year = {2014},
  publisher = {Elsevier}
}

@article{Elkind15,
  title = {Condorcet Winning Sets},
  author = {Elkind, Edith and Lang, J{\'e}r{\^o}me and Saffidine, Abdallah},
  journal = {Social Choice and Welfare},
  volume = {44},
  number = {3},
  pages = {493--517},
  year = {2015},
  publisher = {Springer},
  doi = {10.1007/s00355-014-0853-4},
  url = {https://doi.org/10.1007/s00355-014-0853-4}
}

@inproceedings{AzizLMRW16,
  title = {Strategyproof Peer Selection: Mechanisms, Analyses, and Experiments},
  author = {Aziz, Haris and Lev, Omer and Mattei, Nicholas and Rosenschein, Jeffrey S. and Walsh, Toby},
  booktitle = {Proceedings of the Thirtieth AAAI Conference on Artificial Intelligence},
  pages = {390--396},
  publisher = {AAAI Press},
  year = {2016},
  doi = {10.1609/aaai.v30i1.10038},
  url = {https://doi.org/10.1609/aaai.v30i1.10038}
}

@article{BanikJCMS16,
  title = {Discrete {Voronoi} Games and $\varepsilon$-Nets, in Two and Three Dimensions},
  author = {Banik, Aritra and {De Carufel}, Jean-Lou and Maheshwari, Anil and Smid, Michiel},
  journal = {Computational Geometry},
  volume = {55},
  pages = {41--58},
  year = {2016},
  issn = {0925-7721},
  doi = {10.1016/j.comgeo.2016.02.002},
  url = {https://www.sciencedirect.com/science/article/pii/S0925772116300086}
}

@inproceedings{GkatzelisHS20,
  title = {Resolving the Optimal Metric Distortion Conjecture},
  author = {Gkatzelis, Vasilis and Halpern, Daniel and Shah, Nisarg},
  booktitle = {2020 IEEE 61st Annual Symposium on Foundations of Computer Science},
  pages = {1427--1438},
  publisher = {IEEE},
  year = {2020},
  doi = {10.1109/FOCS46700.2020.00134},
  url = {https://doi.org/10.1109/FOCS46700.2020.00134}
}

@inproceedings{JiangMW20,
  title = {Approximately Stable Committee Selection},
  author = {Jiang, Zhihao and Munagala, Kamesh and Wang, Kangning},
  booktitle = {Proceedings of the 52nd Annual ACM SIGACT Symposium on Theory of Computing},
  pages = {463--472},
  publisher = {ACM},
  year = {2020},
  doi = {10.1145/3357713.3384238},
  url = {https://doi.org/10.1145/3357713.3384238}
}

@article{Rubin22,
  title = {An Improved Bound for Weak Epsilon-Nets in the Plane},
  author = {Rubin, Natan},
  journal = {Journal of the ACM},
  volume = {69},
  number = {5},
  pages = {32:1--32:35},
  year = {2022},
  month = oct,
  publisher = {ACM},
  doi = {10.1145/3555985},
  url = {https://doi.org/10.1145/3555985}
}

@inproceedings{BergW23,
  title = {Improved Bounds for Discrete Voronoi Games},
  author = {{de Berg}, Mark and {van Wordragen}, Geert},
  editor = {Morin, Pat and Suri, Subhash},
  booktitle = {Algorithms and Data Structures},
  pages = {291--308},
  year = {2023},
  publisher = {Springer Nature Switzerland},
  address = {Cham},
  isbn = {978-3-031-38906-1}
}

@inproceedings{CharikarLRVW25,
  title = {Six Candidates Suffice to Win a Voter Majority},
  author = {Charikar, Moses and Lassota, Alexandra and Ramakrishnan, Prasanna and Vetta, Adrian and Wang, Kangning},
  booktitle = {Proceedings of the 57th Annual ACM Symposium on Theory of Computing},
  pages = {1590--1601},
  year = {2025}
}

@misc{LassotaVV26,
  title = {The Condorcet Dimension of Metric Spaces},
  author = {Lassota, Alexandra and Vetta, Adrian and {von Stengel}, Bernhard},
  year = {2026},
  eprint = {2410.09201v5},
  archivePrefix = {arXiv},
  primaryClass = {cs.GT},
  url = {https://arxiv.org/abs/2410.09201v5}
}

@inproceedings{SongNL26,
  title = {A Few Good Choices},
  author = {Song, Haoyu and Nguyen, Thanh and Lin, Young-San},
  booktitle = {Proceedings of the 2026 Annual ACM-SIAM Symposium on Discrete Algorithms (SODA)},
  pages = {4861--4874},
  doi = {10.1137/1.9781611978971.175},
  url = {https://epubs.siam.org/doi/abs/10.1137/1.9781611978971.175},
  eprint = {https://epubs.siam.org/doi/pdf/10.1137/1.9781611978971.175},
  year = {2026}
}

@misc{ZilbersteinBLM26,
  title = {Is Four Enough? Automated Reasoning Approaches and Dual Bounds for Condorcet Dimensions of Elections},
  author = {Zilberstein, Itai and Berker, Ratip Emin and Li, George and Martins, Ruben},
  year = {2026},
  eprint = {2604.19851},
  archivePrefix = {arXiv},
  primaryClass = {cs.GT},
  url = {https://arxiv.org/abs/2604.19851}
}
